\documentclass[11pt,twoside]{article}
\usepackage{amsmath,amssymb,graphicx,enumitem,mathtools,subcaption,comment,multirow,blkarray,booktabs,bm,float,xcolor}
\usepackage{amsthm,enumitem,array}
\usepackage{longtable}
\usepackage{url}
\usepackage{breakurl}
\usepackage[colorlinks=true,citecolor=blue,linkcolor=blue,urlcolor=blue,bookmarks,bookmarksopen,bookmarksdepth=2,backref=page,breaklinks]{hyperref}
\usepackage{bbm}

\DeclareCaptionSubType * [alph]{table}
\usepackage{tabularx}
\usepackage{lipsum}
\newcolumntype{L}[1]{>{\raggedright\arraybackslash}p{#1}}
\newcolumntype{C}[1]{>{\centering\arraybackslash}p{#1}}
\newcolumntype{R}[1]{>{\raggedleft\arraybackslash}p{#1}}

\usepackage[labelsep=period]{caption}

\usepackage[a4paper]{geometry}
\makeatletter
\renewcommand\title[1]{\gdef\@title{\reset@font\Large\bfseries #1}}
\renewcommand\section{\@startsection {section}{1}{\z@}%
	{-3.5ex \@plus -1ex \@minus -.2ex}%
	{2.3ex \@plus.2ex}%
	{\normalfont\large\bfseries}}
\renewcommand\subsection{\@startsection{subsection}{2}{\z@}%
	{-3ex\@plus -1ex \@minus -.2ex}%
	{1.5ex \@plus .2ex}%
	{\normalfont\normalsize\bfseries}}
\renewcommand\subsubsection{\@startsection{subsubsection}{3}{\z@}%
	{-2.5ex\@plus -1ex \@minus -.2ex}%
	{1.5ex \@plus .2ex}%
	{\normalfont\normalsize\bfseries}}

\def\@runningauthor{}\newcommand{\runningauthor}[1]{\def\runningauthor{#1}}
\def\@runningtitle{}\newcommand{\runningtitle}[1]{\def\runningtitle{#1}}

\renewcommand{\ps@plain}{%
	\renewcommand{\@evenhead}{\footnotesize\scshape \hfill\runningauthor\hfill}
	\renewcommand{\@oddhead}{\footnotesize\scshape \hfill\runningtitle\hfill}}
\g@addto@macro\bfseries{\boldmath}

\makeatother

\setlist{leftmargin=*, itemsep=2pt, topsep=2pt, parsep=0pt, partopsep=5pt}
\usepackage{tikz,color} 
\usetikzlibrary{shapes,arrows,calc,matrix,positioning,fit}

\DeclareMathOperator{\dev}{dev}

\DeclareMathOperator{\Aut}{Aut}

\DeclareMathOperator{\wt}{wt}

\DeclareMathOperator{\supp}{supp}

\theoremstyle{plain}
\newtheorem{theorem}{Theorem}[section]
\newtheorem{lemma}[theorem]{Lemma}
\newtheorem{proposition}[theorem]{Proposition}
\newtheorem{result}[theorem]{Result}

\newtheorem{corollary}[theorem]{Corollary}

\theoremstyle{definition}

\newtheorem{hypothesis}[theorem]{Hypothesis}
\newtheorem{example}[theorem]{Example}

\newtheorem{definition}[theorem]{Definition}
\newtheorem{remark}[theorem]{Remark}

\newtheorem{question}[theorem]{Question}

\numberwithin{theorem}{section}
\numberwithin{equation}{section}
\numberwithin{table}{section}
\numberwithin{figure}{section}

\newcommand{\F}{\mathbb F}
\newcommand{\Z}{\mathbb Z}
\newcommand{\N}{\mathbb N}
\newcommand{\VF}{\mathcal{VF}}
\newcommand{\NF}{\mathcal{NF}}
\DeclareMathOperator{\nl}{nl}

\usepackage{bookmark}
\bookmarksetup{
	numbered, 
	open,
}

\newcommand{\aBold}{\mathbf{a}}
\newcommand{\bBold}{\mathbf{b}}
\newcommand{\cBold}{\mathbf{c}}

\newcommand{\uBold}{\mathbf{u}}
\newcommand{\vBold}{\mathbf{v}}
\newcommand{\wBold}{\mathbf{w}}
\newcommand{\xBold}{\mathbf{x}}
\newcommand{\yBold}{\mathbf{y}}

\newcommand{\oBold}{\mathbf{0}}

\newcommand{\C}{\mathcal{C}}
\newcommand{\D}{\mathbb D}

\newcommand{\graphElt}[2]{\begin{pmatrix}
#1\\
#2
\end{pmatrix}}

\newcommand*\colvec[3][]{
    \begin{pmatrix}\ifx\relax#1\relax\else#1\\\fi#2\\#3\end{pmatrix}
}

\newcommand{\graph}[1]{\mathcal{G}_{#1}}
\newcommand{\support}[1]{\mathcal{D}_{#1}}

\renewcommand*{\backref}[1]{}
\renewcommand*{\backrefalt}[4]{%
	\ifcase #1 (Not cited.)%
	\or        (Cited on page~#2.)%
	\else      (Cited on pages~#2.)%
	\fi}

\let\OLDthebibliography\thebibliography
\renewcommand\thebibliography[1]{
	\OLDthebibliography{#1}
	\setlength{\parskip}{0pt}
	\setlength{\itemsep}{0pt plus 0.3ex}
}
\title{Linear codes and incidence structures\\of bent functions and their generalizations}
\runningtitle{Linear codes and incidence structures of bent functions and their generalizations}

\author{
$^1$Wilfried Meidl \qquad $^2$Alexandr A. Polujan \qquad $^2$Alexander Pott\\[2mm]
\small $^1$Institut f\"{u}r Mathematik, Alpen-Adria-Universit\"{a}t Klagenfurt\\
\small Universit\"{a}tsstra{\ss}e 65-67, 9020, Klagenfurt, Austria\\
\small {\tt meidlwilfried@gmail.com}\\[2mm]
\small $^2$Faculty of Mathematics\\
\small Institute of Algebra and Geometry\\
\small Otto von Guericke University\\
\small Universit\"{a}tsplatz 2, 39106, Magdeburg, Germany\\
\small\tt alexandr.polujan@$\{$gmail.com,ovgu.de$\}$, alexander.pott@ovgu.de
}

\runningauthor{W. Meidl, A. A. Polujan, A. Pott}

\date{}

\begin{document}
	\maketitle
	\thispagestyle{empty}
	\begin{abstract}
		\noindent
		In this paper, we consider further applications of $(n,m)$-functions for the construction of 2-designs. For instance, we provide a new application of the extended Assmus-Mattson theorem, by showing that linear codes of certain APN functions with the classical Walsh spectrum support 2-designs. With this result, we give several sufficient conditions for an APN function with the classical Walsh spectrum to be CCZ-inequivalent to a quadratic one. On the other hand, we use linear codes and combinatorial designs in order to study important properties of $(n,m)$-functions. In particular, we provide a characterization of a quadratic Boolean bent function by means of the 2-transitivity of its automorphism group.  Finally, we give a new design-theoretic characterization of $(n,m)$-plateaued and $(n,m)$-bent functions and provide a coding-theoretic as well as a design-theoretic interpretation of the extendability problem for $(n,m)$-bent functions.
		
		 
		% 05B10 Combinatorial aspects of difference sets
		% 06E30 Boolean functions
		% 14G50 Applications to coding theory and cryptography of arithmetic geometry
		% 94C30 Applications of design theory to circuits and networks
		% 94B05 Linear codes, general
	\end{abstract}	
	\section{Introduction}
	\label{section: 1 Introduction}
	Cryptographically significant Boolean and vectorial functions have many different applications in coding theory and design theory and vice versa. The classical application is the construction of linear codes and $t$-designs with optimal parameters from classes of Boolean and vectorial functions having nice cryptographic properties (e.g., high nonlinearity or low differential uniformity). At the same time, linear codes and incidence structures, constructed from Boolean and vectorial functions, are often used to characterize certain classes of functions or to distinguish different ones. Recently, there has been a lot of work on $(n,m)$-functions, their designs and codes~\cite{Ding19BCD,Ding2020,Li20VanishingFlats,polujan2020design,TangDing2019}. In this paper, we continue the study of the interaction between Boolean and vectorial functions, linear codes and $t$-designs. We begin with a thorough introduction into the subject.
	
	\subsection{Preliminaries}
	\paragraph{Boolean and vectorial Boolean functions.}	
	Let $\F_2=\{0,1\}$ be the finite field with two elements and let $\F_2^n$ be the vector space of dimension $n$ over $\F_2$. Mappings $F\colon\F_2^n\rightarrow\F_2^m$ are called \emph{$(n,m)$-functions}, in particular,  $(n,1)$-functions are referred to as \emph{Boolean} functions, while $(n,m)$-functions with $m\ge2$ are referred to as \emph{vectorial} functions. Any vectorial function $F\colon\F_2^n\rightarrow\F_2^m$ can be uniquely (up to the choice of basis of $\F_2^m$) associated with $m$ \emph{coordinate} Boolean functions $f_i\colon\F_2^n\rightarrow\F_2$ for $1\le i\le m$ as a column-vector $F(\mathbf{x}):=(f_1(\mathbf{x}),\ldots,f_m(\mathbf{x}))^T$.  \emph{Component functions} of an $(n,m)$-function $F$ are Boolean functions $F_{\bBold}\colon \xBold\in\F_2^n \mapsto \langle \bBold, F(\xBold) \rangle_{m}$, where $\bBold\in\F_2^m$ and $\langle \cdot,\cdot \rangle_m$ is a non-degenerate bilinear form on $\F_2^m$. Any Boolean function $f\colon\F_2^n\rightarrow\F_2$ has a unique multivariate polynomial representation in the ring $\F_2[x_1,\dots,x_n]/(x_1\oplus x_1^2,\dots,x_n\oplus x_n^2)$, called the \emph{algebraic normal form} (and denoted by ANF for short) and given by $f(\mathbf{x})=\bigoplus_{\mathbf{v}\in\F_2^n}c_{\mathbf{v}} \left( \prod_{i=1}^{n} x_i^{v_i} \right)$, where $\xBold = (x_1,\dots, x_n)\in\F_2^n$,  $c_{\vBold}\in\F_2$ and $\vBold = (v_1,\dots, v_n)\in\F_2^n$. The \emph{algebraic degree of a Boolean function} $f$ on $\F_2^n$ is defined as the algebraic degree of its ANF as a multivariate polynomial and denoted by $\deg(f)$. The \emph{algebraic degree} of a vectorial $(n,m)$-function $F$ is defined as the maximum algebraic degree of its coordinate functions, that is, $\deg(F):=\max_{1\le i\le m} \deg(f_i)$. Clearly, the algebraic degree of an $(n,m)$-function $F$ is at most $n$. The \emph{set of affine $(n,m)$-functions}, denoted by $\mathcal{A}_{n,m}$ in the vectorial and by $\mathcal{A}_{n}$ in the Boolean case, respectively, is the set of $(n,m)$-functions of algebraic degree at most one, that is, $\mathcal{A}_{n,m}:=\{A \colon \F_2^n\rightarrow \F_2^m \mid \deg(A) \le 1 \}$. Clearly, any affine Boolean function $a$ on $\F_2^n$ can be represented as $a(\xBold):=\l_\aBold(\xBold)\oplus b$ for some constant $b\in\F_2$ and a linear function $l_\aBold$ on $\F_2^n$, given by $l_\aBold(\xBold):=\langle \aBold, \xBold \rangle_{n}$.
	
	Affine functions, being simple algebraic objects, have to be avoided in order to construct secure cryptographic systems. However, it is not enough to take an $(n,m)$-function having only high algebraic degree, it must also satisfy several cryptographic criteria in order to be considered as a component of block ciphers. Among them are high nonlinearity and low differential uniformity, which represent a measure of being different from the set of all affine functions and formally are defined as follows.  
	
	We endow the set of Boolean functions $\mathfrak{B}_n$ on $\F_2^n$ with the structure of a metric space $(\mathfrak{B}_n, d)$, where $d(f, g):=|\{ \xBold\in\F_2^n\colon f(\xBold)\neq g(\xBold) \}|$ is the \emph{Hamming distance} between Boolean functions $f,g\in\mathfrak{B}_n$. The \emph{nonlinearity of a Boolean function} $f$ on $\F_2^n$, denoted by $\nl(f)$, is a measure of distance between the function $f$ and the set of all affine functions $\mathcal{A}_n$, namely $\nl(f):=\min\limits_{a\in \mathcal{A}_n} d(f,a)$. The definition of nonlinearity can be extended for the vectorial case using the notion of component functions as follows. The \emph{nonlinearity of a vectorial $(n,m)$-function} $F$ is the minimum nonlinearity of all its component functions and is given by $\nl(F):=\min\limits_{a\in\mathcal{A}_n,\bBold\in\F_2^m \setminus \{ \mathbf{0}\} } d(F_{\bBold},a)$. The nonlinearity of an $(n,m)$-function $F$ is at most $\nl( F ) \le 2 ^ { n - 1 } - 2 ^ { \frac { n } { 2 }-1}$ and $(n,m)$-functions, achieving this bound are called \emph{perfect nonlinear}.	
	\begin{definition}
		An $(n,m)$-function $F$ is called \emph{perfect nonlinear} if $\nl( F ) = 2 ^ { n - 1 } - 2 ^ { \frac { n } { 2 } - 1 }$.
	\end{definition}
	
	The standard tool to compute the nonlinearity of an $(n,m)$-function is the \emph{Walsh transform} $\hat{\chi}_F\colon \F_2^n\times\F_2^m\rightarrow \Z$, which is defined in the following way
	\begin{equation*}
		\hat{\chi}_F(\aBold,\bBold):=\hat{\chi}_{F_{\bBold}}(\aBold) \quad \mbox{and} \quad \hat{\chi}_{F_{\bBold}}(\aBold):=\sum\limits_{\xBold\in\F_2^n} (-1)^{F_{\bBold}(\xBold)\oplus\langle \aBold, \xBold \rangle_n}\;\mbox{ for }\; \aBold\in\F_2^n\;\mbox{ and }\bBold\in\F_2^m.
	\end{equation*}
  	Using the Walsh transform, the nonlinearity of an $(n,m)$-function $F$ can be computed as $\nl(F):=2^{n-1}- \frac{1}{2}\cdot\max\limits_{\aBold\in\F_2^n,\bBold\in\F_2^m \setminus \{ \mathbf{0}\} }\left| \hat{\chi}_F(\aBold,\bBold) \right|$. The multiset $\Lambda_{F}:=\{* \hat{\chi}_F(\aBold,\bBold) \colon \aBold\in\F_2^n,\bBold\in\F_2^m\setminus\{\mathbf{0}\} *\}$ is called the \emph{Walsh spectrum} of an $(n,m)$-function $F$.
  	
  	In order to introduce the differential uniformity, we define the \emph{first-order derivative} of an $(n,m)$-function $F$, that is a mapping $D_{\aBold}F$ given by $D_{\aBold}F\colon\xBold\mapsto F(\xBold\oplus \aBold)\oplus F(\xBold)$.  
An $(n, m)$-function $F$ has {\emph{differential uniformity} $\delta$ if the value $\delta(F)$, defined as follows
  	\begin{equation*}
  	\delta(F):=\max\limits_{\aBold\in\F_2^n\setminus \{ \mathbf{0}\},\bBold\in\F_2^m} \delta_F(\aBold,\bBold), \quad \mbox{where} \quad \delta_F(\aBold,\bBold):=|\{\xBold \in \F_2^n \colon D_{\aBold}F(\xBold)=\bBold \}|,
  	\end{equation*}
  	is equal to $\delta$, (see \cite{Nyberg94}).} The multiset set $\Delta_F:=\{* \delta_F(\aBold,\bBold) \colon \aBold \in \F_2^n\setminus\{\mathbf{0}\}, \bBold \in \F_2^m *\}$ is called the \emph{differential spectrum} of the function $F$. The differential uniformity of an $(n,m)$-function $F$ is at least $\delta(F)\ge 2^{n-m}$.
	
	In general, $(n,m)$-functions with the highest nonlinearity and $(n,m)$-functions with the lowest differential uniformity are two different sets of functions. However, as the following result shows, in some cases these sets are the same. 
		
	\begin{result}
		Let $F$ be an $(n,m)$-function with $n$ even and $m\le n/2$. The following statements are equivalent.
		\begin{enumerate}
			\item $F$ is perfect nonlinear, that is $\nl(F)=2^{n-1}-2^{\frac{n}{2}-1}$.
			\item The differential uniformity of $F$ is $\delta(F)=2^{n-m}$.
			\item For all $\aBold\in\F_2^n$ and $\bBold\in\F_2^m\setminus \{ \mathbf{0} \}$ the Walsh transform satisfies $|\hat{\chi}_F(\aBold,\bBold)|=2^{n/2}$.
		\end{enumerate}
	\end{result}
	\begin{remark}
		The last statement means, that all the nonzero component functions $F_\bBold$ of a vectorial perfect nonlinear $(n,m)$-function $F$ are \emph{Boolean bent functions} on $\F_2^n$, that is their Walsh transform satisfies $\hat{\chi}_{F_\bBold}(\aBold)=\pm 2^{n/2}$ for all $\aBold\in\F_2^n$. For this reason, vectorial perfect nonlinear $(n,m)$-functions are also called \emph{vectorial bent functions}. The conditions on $n$ and $m$ in the formulation of the previous result are explained by the fact, that perfect nonlinear $(n,m)$-functions exist if and only if $n$ is even and $m\le n/2$, see~\cite{ROTHAUS1976300} and~\cite{Nyberg91}, respectively.
\end{remark}
	
	Clearly, any $(n,m+1)$-bent function $G$ can be written as $G(\xBold)=(F(\xBold),f(\xBold))^T$ for some $(n,m)$-bent function $F$ and a Boolean bent function $f$ on $\F_2^n$. However, it is not clear, whether the converse is also true, that is whether an arbitrary $(n,m)$-bent function $F$ could be extended by a Boolean bent function $f$ on $\F_2^n$ to an $(n,m+1)$-bent function $G\colon \xBold\mapsto(F(\xBold),f(\xBold))^T$. This question motivates the following definition.
	\begin{definition}\label{definition: lonely bent functions}
		An $(n,m)$-bent function $F$ is called \emph{extendable}, if there exists a Boolean bent function $f\colon\F_2^n\rightarrow\F_2$, such that the function $G\colon \xBold\in\F_2^n\mapsto \left(F(\xBold),f(\xBold)\right)^T$ is $(n,m+1)$-bent. If no such a bent function $f$ exists, the function $F$ is called \emph{non-extendable} or \emph{lonely}, see~\cite{polujan2020design}.
	\end{definition}
		
	Perfect nonlinearity of $(n,m)$-functions for $n$ even and $m\le n/2$ is characterized by  the minimality of either the value set of the Walsh transform or of the differential spectrum. Further generalizations of perfect nonlinear functions are obtained by relaxing slightly the minimality conditions.
	
	\begin{definition}
		A Boolean function $f\colon\F_{2}^{n}\to \F_{2}$ is said to be \emph{$s$-plateaued} if for all $\aBold\in\F_2^n$ the absolute value of its Walsh transform takes only two values, i.e.,  $|\hat{\chi}_f(\aBold)|\in\{0,2^{\frac{n+s}{2}}\}$. The value $2^{\frac{n+s}{2}}$ is called the \emph{amplitude} of an $s$-plateaued Boolean function $f$. An $(n,m)$-function $F$ is said to be \emph{$s$-plateaued} if all its component functions $F_{\bBold}$ with $\bBold\neq\mathbf{0}$ are $s$-plateaued. If all the component functions $F_{\bBold}$ of an $(n,m)$-function $F$ are $s_\bBold$-plateaued (not necessarily with the same amplitude), then $F$ is called an \emph{$(n,m)$-plateaued} function. Boolean 1-plateaued  functions on $\F_2^n$ with $n$ odd and 2-plateaued functions on $\F_2^n$ with $n$ even are called \emph{semi-bent}.
	\end{definition}
	\begin{definition}
		An $(n,m)$-function $F$ is called \emph{differentially two-valued}, if there are only two different values in the differential spectrum, that is $\Delta_F=\{0,2^s\}$ (multiplicities are omitted). In particular, $(n,n)$-functions $F$ with $\Delta_F=\{0,2\}$ are called \emph{almost perfect nonlinear} or simply \emph{APN}.
	\end{definition}

	On the set of all $(n,m)$-functions we introduce the following equivalence relations. We say that two $(n,m)$-functions $F,F'$ are:
	\begin{itemize}
		\item \emph{Extended-affine equivalent} (\emph{EA-equivalent} for short), if there exist a linear permutation $A_1$ of $\F_{2}^{m}$, an affine permutation $A_2$ of $\F_{2}^{n}$ and an affine function $A_3\colon\F_2^{n}\rightarrow\F_2^m$ such that $F=A_1\circ F' \circ A_2 \oplus A_3$;
		\item \emph{Carlet-Charpin-Zinoviev-equivalent}, or simply \emph{CCZ-equivalent}, if there exists an affine permutation $\mathcal{L}$ of $\F_{2}^{n} \times \F_{2}^{m}$ such that $\mathcal{L}\left(\graph{F}\right)=\graph{F'}$, where $\graph{F}:=\left\{(\xBold, F(\xBold)) \colon \xBold \in \F_{2}^{n}\right\}$ is the \emph{graph} of an $(n,m)$-function $F$.
	\end{itemize}	
	Note that in general CCZ-equivalence is a coarser equivalence relation than EA-equivalence, however, for Boolean functions and $(n,m)$-bent functions they coincide~\cite{BudaghyanC2009,KyureghyanP2008}.

	\paragraph{Linear codes and incidence structures.}
	A \emph{linear code} $\C$ over $\F_2$ is a vector subspace $\C \subseteq \F_2^v$. Elements of a linear code $\mathcal{C} \subseteq \F_2^v$ are called \emph{codewords} and said to have a length $v$. The number of nonzero coordinates of a codeword $\cBold\in\mathcal{C}$ is called the \emph{weight} of $\cBold$ and is denoted by $\wt(\cBold)$.  The \emph{support} of a codeword $\cBold=(c_1,\ldots,c_v)\in\mathcal{C}$ is defined by $\supp(\cBold)=\left\{1 \leq i \leq v: c_{i} \neq 0\right\} \subseteq\{1,2,3, \ldots, v\}$.  Further we will denote by $A_i$ the number of codewords of weight $i$ in the code $\mathcal{C}$. For a linear code $\mathcal{C}\subseteq \F_2^v$ we will call the polynomial $W_{\mathcal{C}}(z):=\sum_{i=0}^{v}A_w z^i$ the \emph{weight enumerator} of $\mathcal{C}$. The \emph{minimum distance} of a linear code is the minimum weight of its nonzero codewords. We also say, that a linear code $\mathcal{C}\subseteq \F_2^v$ is a $[v,k]$-\emph{linear} code, if $\mathcal{C}$ has dimension $k$ and a $[v,k,d]$-\emph{linear} code, if $\mathcal{C}$ is a $[v,k]$-linear code, which has the minimum distance $d$. An integer $\rho=\rho(\mathcal{C})$ is said to be the \emph{covering radius} of the linear code $\mathcal{C}$ of length $v$, if $\rho=\max\limits_{\xBold\in\F_2^v}\min\limits_{\cBold\in\mathcal{C}}d(\xBold, \cBold)$. The \emph{dual} of the $[v,k]$-linear $\mathcal{C}$ code is the $[v,v-k]$-linear code $\mathcal{C}^\perp$, defined by $\mathcal{C}^\perp:=\{ \mathbf{u}\in\F_2^v\colon \mathbf{u} \cdot \mathbf{w}= 0\mbox{ for all } \mathbf{w}\in\mathcal{C} \}$, where $ \mathbf{u} \cdot \mathbf{w}=u_1w_1\oplus u_2w_2\oplus\cdots\oplus u_nw_n$ is the standard dot product on $\F_2^n$.
	
	An \emph{incidence structure} is a pair $\mathcal{S}=(\mathcal{P}, \mathcal{B})$, where $\mathcal{P}$ is a finite set, called the \emph{point set} of $\mathcal{S}$ and $\mathcal{B}$ is the collection of subsets of $\mathcal{P}$, called the \emph{block set} of $\mathcal{S}$. All the information about an incidence structure $\mathcal{S}$ is contained in its \emph{incidence matrix} $\mathbf{M}(\mathcal{S})=(m_{i,j})$, which is a binary $b\times v$ matrix with $m_{i,j}=1$ if $p_j\in B_i$ for $p_j\in\mathcal{P},B_i\in\mathcal{B}$ and $m_{i,j}=0$ otherwise. Two incidence structures $\mathcal{S}$ and $\mathcal{S}'$ are \emph{isomorphic}, if there exist permutation matrices $\mathbf{P}$ and $\mathbf{Q}$ such that $\mathbf{M}(\mathcal{S})=\mathbf{P}\cdot \mathbf{M}(\mathcal{S}') \cdot \mathbf{Q}$. An incidence structure $\mathcal{D}=(\mathcal{P}, \mathcal{B})$ is called a $t$-$(v, k, \lambda)$ \emph{design}, if the cardinality of the point set $\mathcal{P}$ is $v$, the block set $\mathcal{B}$ is a collection of \emph{$k$-subsets} of $\mathcal{P}$ and every $t$-subset of points of $\mathcal{P}$ is contained in exactly $\lambda$ blocks of $\mathcal{B}$. The parameter $\lambda$ is called the \emph{valency}. Any $t$-$(v, k, \lambda)$ design $\mathcal{D}$ is a \emph{regular} incidence structure, i.e., any point of $\mathcal{D}$ is contained in the same number of blocks $r$, which is called the \emph{replication number} of $\mathcal{D}$.  In the case $\mathcal{B}=\varnothing$ one sets $\lambda=0$ and calls $(\mathcal{P}, \varnothing)$ a $t$-$(v, k, 0)$ design for any $t$ and $k$ with $1 \leq t \leq v$ and $0 \leq k \leq v$. Further we will call such a design \emph{trivial}. In the case of $1$-$(v, k, \lambda)$ designs the parameters $r$ and $\lambda$ coincide, for this reason we will call $1$-designs regular incidence structures. If two 1-designs have the same replication number $r$, we will call them \emph{equiregular}. We say that a $t'$-$(v',k,\lambda')$ design $D' = (\mathcal{P}',\mathcal{B}')$ is a \emph{subdesign} of a $t$-$(v,k,\lambda)$ design $D=(\mathcal{P},\mathcal{B})$, if $\mathcal{P}'\subseteq \mathcal{P}$ and $\mathcal{B}'\subseteq \mathcal{B}$. Finally, if for a design $D=(\mathcal{P},\mathcal{B})$ there exist $n$ subdesigns $D_i=(\mathcal{P},\mathcal{B}_i)$ such that $\mathcal{B}=\bigsqcup_{i=1}^n \mathcal{B}_i$, we say that that $D$ \emph{is partitioned into subdesigns} $D_1,\ldots, D_n$ and write $D=\bigsqcup_{i=1}^n D_i$.
	
	\subsection{Classical tools to construct $t$-designs from linear codes}
	Let $\mathcal{C}$ be a linear code of length $v$. Let $w$ be an integer, such that $A_w\neq 0$. Define the following two sets $\mathcal{P}(\mathcal{C}):=$ $\{1, \ldots, v \}$, i.e., the set of all possible coordinate positions of $\mathcal{C}$ and $\mathcal{B}_{w}(\mathcal{C}):=\{ \supp(\cBold) : \wt(\cBold)=w \mbox{ and } \mathbf{c} \in \mathcal{C} \}$ i.e., the collection of the supports of the codewords of the fixed weight $w$. If $\left(\mathcal{P}(\mathcal{C}), \mathcal{B}_{w}(\mathcal{C})\right)$ is a $t$-design for a fixed $w$, we say that the codewords of the weight $w$ in the code $\mathcal{C}$ \emph{hold a $t$-design}. If for any $0 \leq w \leq v$ the codewords of the weight $w$ hold $t$-designs we say that the code $\mathcal{C}$ \emph{supports $t$-designs}. Note that this definition is also valid even if $\left(\mathcal{P}(\mathcal{C}), \mathcal{B}_{w}(\mathcal{C})\right)$ is a trivial design.
	
	In general, it is a nontrivial problem to construct $t$-designs. One of the standard ways to construct $t$-designs from linear codes is to consider the supports of the codewords of a fixed weight, and check whether either the automorphism group of the given code is $t$-transitive, or the conditions of the original Assmus-Mattson theorem are fulfilled. Below we formulate these results for binary linear codes, since these are the only codes considered in this paper, however the original statements are valid for linear codes over $\F_{q}$ with $q$ being a prime power.

	Let $G$ be a group, $S$ be a set and let $\cdot\colon G\times S\to S$ be a group action. The group action is called \emph{$t$-transitive} if for any two ordered $t$-tuples $(x_1,\ldots,x_t)$ and $(y_1,\ldots,y_t)$ of pairwise distinct elements of $S$, there exists an element $g\in G$ such that for all $i\in \{1,\ldots,t\}$ holds $g\cdot x_i=y_i$.
		
	\begin{result}[Transitivity theorem]\label{result: Designs from nice automorphism groups}\cite{assmus1992designs}.
		Let $\mathcal{C}$ be a code of length $v$ over $\F_2$ where $\Aut(\mathcal{C})$ is $t$-transitive. Then the codewords of any weight $w \ge t$ of $\mathcal{C}$ hold a $t$-design.
	\end{result}
	
	\begin{result}[Original Assmus-Mattson Theorem]\cite[Theorem 2.2]{TangDing2019}\label{result: Original Assmus-Mattson theorem}
		Let $\mathcal{C}$ be a linear code over $\F_2$ of length $v$ and minimum distance $d$. Let $\mathcal{C}^{\perp}$ with minimum weight $d^{\perp}$ denote the dual code of $\mathcal{C}$. Let $t$ with $1 \leq t<\min \left\{d, d^{\perp}\right\}$ be an integer such that there are at most $d^{\perp}-t$ weights of $\mathcal{C}$ in $\{1,2, \ldots, v-t\}$. Then $\left(\mathcal{P}(\mathcal{C}),\mathcal{B}_{k}(\mathcal{C})\right)$ and $\left(\mathcal{P}\left(\mathcal{C}^{\perp}\right), \mathcal{B}_{k}\left(\mathcal{C}^{\perp}\right)\right)$ are $t$-designs for all $k\in\{0,1, \ldots, v\}$.
	\end{result}
	
	This is the recent extension of the original Assmus-Mattson Theorem, which was shown to be a powerful tool one may use to construct 2-designs from cryptographically significant Boolean and vectorial functions.
	
	\begin{result}[Extended Assmus-Mattson Theorem]\label{result: Extended Assmus-Mattson theorem}\cite[Theorem 5.3]{TangDing2019}
		Let $\mathcal{C}$ be a linear code over $\F_2$ with length $v$ and minimum weight $d$. Let $\mathcal{C}^{\perp}$ denote the dual code of $\mathcal{C}$ with minimum weight $d^{\perp}$. Let $s$ and $t$ be two positive integers such that $t$ satisfies $t<\min\left\{d, d^{\perp}\right\}$. Let $S$ be an $s$-subset of $\{d, d+1, \ldots, v-t\}$. Suppose that $\left(\mathcal{P}(\mathcal{C}), \mathcal{B}_{\ell}(\mathcal{C})\right)$ and
		$\left(\mathcal{P}\left(\mathcal{C}^{\perp}\right), \mathcal{B}_{\ell^{\perp}}\left(\mathcal{C}^{\perp}\right)\right)$ are $t$-designs for $\ell \in\{d, d+1, \ldots, v-t\} \setminus S$ and $0 \leq \ell^\perp \leq s+t-1$. Then $\left(\mathcal{P}(\mathcal{C}), \mathcal{B}_{k}(\mathcal{C})\right)$ and
		$\left(\mathcal{P}\left(\mathcal{C}^{\perp}\right), \mathcal{B}_{k}\left(\mathcal{C}^{\perp}\right)\right)$ are $t$-designs for any $t \leq k \leq v$. 
	\end{result}	
	
	For an $(n,m)$-function $F$, we define the linear code $\mathcal{C}_F$ as the linear code, generated by the rows of the following matrix
	\begin{equation}
	\begin{bmatrix}
	1\\
	\xBold\\
	F(\xBold)
	\end{bmatrix}_{\xBold\in\F_2^n}.
	\end{equation}
	As we show further, this code and its dual contain all information about the Walsh- and differential spectra of a given $(n,m)$-function $F$. Moreover, they can be used to distinguish inequivalent functions, as the following statement shows.
	\begin{result}\cite[Theorem 9]{EdelP09}
		Two $(n,m)$-functions $F$ and $F'$ are CCZ-equivalent iff the linear codes $\mathcal{C}_F$ and $\mathcal{C}_{F'}$ (or equivalently $\mathcal{C}^\perp_F$ and $\mathcal{C}^\perp_{F'}$) are permutation equivalent.
	\end{result}
	Clearly, not only the weight enumerators of CCZ-equivalent functions are invariants, but also the incidence structures, supported by the codewords of fixed weight.
	\begin{result}\label{theorem: Support designsare  CCZ-invariant}
		Let $F$ and $F'$ be two CCZ-equivalent $(n,m)$-functions. Then the incidence structures $\left(\mathcal{P}(\mathcal{C}_F), \mathcal{B}_{k}(\mathcal{C}_F)\right)$
		and $\left(\mathcal{P}(\mathcal{C}_{F'}), \mathcal{B}_{k}(\mathcal{C}_{F'})\right)$ as well as 
		$\left(\mathcal{P}\left(\mathcal{C}^{\perp}_{F}\right), \mathcal{B}_{l}\left(\mathcal{C}^{\perp}_{F}\right)\right)$ and $\left(\mathcal{P}\left(\mathcal{C}^{\perp}_{F'}\right), \mathcal{B}_{l}\left(\mathcal{C}^{\perp}_{F'}\right)\right)$ are isomorphic for all $0\le k,l\le 2^n$.
	\end{result}
	Finally, we give the connection between the number of codewords of the weight 4 in the linear code $\mathcal{C}^\perp_F$ of an $(n,m)$-function $F$ and the fourth power moments of the Walsh transform.
	\begin{result}\cite[Theorem 2.5.]{Arshad2018}
	Let $F$ be an $(n,m)$-function. Then the number of codewords of the weight 4 in $\mathcal{C}^\perp_F$ is given by	
	\begin{equation}\label{equation: Number of weight 4 codewords in CF dual}
	A_4=\dfrac{1}{24}\left( \dfrac{1}{2^{n+m}} \left(\sum\limits_{\aBold\in\F_2^n,\bBold\in\F_2^m}(\hat{\chi}_F(\aBold,\bBold))^4 \right)-3\cdot 2^{2n} +2^{n+1}\right).
	\end{equation}
	\end{result}
	Using this statement it is not difficult to derive the following characterizations of perfect and almost perfect nonlinear functions.
	\begin{corollary}\label{corollary: Coding-theoretic characterization of bent functions}
		Let $n=2k$. The following statements are equivalent.
		\begin{enumerate}
			\item An $(n,m)$-function $F$ is bent.
			\item The linear code $\mathcal{C}_F$ is a $[2^n,n+m+1,2^{n-1}-2^{k-1}]$-linear code with the weight enumerator
			\begin{equation}\label{equation: Bentness via CF}
			W_{\mathcal{C}_F}(z)=1+\left(2^{m}-1\right) 2^{n} z^{2^{n-1}-2^{k-1}}+(2^{n+1}-2) z^{2^{n-1}} + \left(2^{m}-1\right) 2^{n} z^{2^{n-1}+2^{k-1}}+z^{2^{n}}.
			\end{equation}
			\item The linear code $\mathcal{C}^\perp_F$ is a $[2^n,2^n-n-m-1,4]$-linear code with the number of weight 4 codewords given by
			\begin{equation}\label{equation: Bentness via CF dual}
			A_4=\dfrac{1}{3}\left(2^{3n-m-3}-2^{2n-m-3}-2^{2n-2}+2^{n-2}  \right),
			\end{equation}
			which is the minimum possible value for an $(n,m)$-function $F$ with $n$ even and $m\le n/2$.
		\end{enumerate}		
	\end{corollary}
	
	\begin{corollary}\label{corollary: Coding-theoretic characterization of APN functions}
		An $(n,n)$-function $F$ is APN if and only if $\mathcal{C}^\perp_F$ is a $[2^n,2^n-2n-1,6]$-linear code or, equivalently, if the number of weight 4 codewords in $\mathcal{C}^\perp_F$ is $A_4=0$.
	\end{corollary}
	\subsection{Motivation}
	In this paper, we extend the recent works~\cite{Li20VanishingFlats,TangDing2019}, which studied incidence structures, arising from the linear codes $\mathcal{C}_F$ and $\mathcal{C}_F^\perp$ of $(n,n)$-functions $F$. In particular, Li et al.~\cite{Li20VanishingFlats}, motivated by the study of CCZ-inequivalence of $(n,n)$-functions $F$, introduced the partial quadruple system $\VF(F)$, called the \emph{vanishing flats} of the $(n,n)$-function $F$, which captures a detailed combinatorial information about the given function. Formally, it is defined in the following way: $\VF(F)=(\mathcal{P},\VF_F)$, where the point set is given by $\mathcal{P}=\{ \xBold \colon \xBold \in\F_2^n \}$ and the block set $\VF_F$ is given as follows
	\begin{equation}\label{equation: Block set of the vanishing flats}
	\VF_F= \left\{ \{ \xBold_1,\xBold_2,\xBold_3,\xBold_4  \} \colon \bigoplus\limits_{i=1}^{4} \graphElt{\xBold_i}{F(\xBold_i)}=\graphElt{\mathbf{0}}{\mathbf{0}} \mbox{ for } \xBold_i\in\F_2^n  \right\} .
	\end{equation}
	It is not difficult to see, that the incidence structure $\VF(F)$ is $\left(\mathcal{P}\left(\mathcal{C}^{\perp}_{F}\right), \mathcal{B}_{4}\left(\mathcal{C}^{\perp}_{F}\right)\right)$, also studied by Tang, Ding and Xiong~\cite{TangDing2019}, who mainly focused on the construction of 2-designs of the form  $\left(\mathcal{P}\left(\mathcal{C}_{F}\right), \mathcal{B}_{k}\left(\mathcal{C}_{F}\right)\right)$ and $\left(\mathcal{P}\left(\mathcal{C}^{\perp}_{F}\right), \mathcal{B}_{k}\left(\mathcal{C}^{\perp}_{F}\right)\right)$ from $(n,n)$-functions $F$ having nice cryptographic properties. In this way, the number of vanishing flats $|\VF_F|$ is equal to the number of weight four codewords $A_4$ of $\mathcal{C}^{\perp}_{F}$ and given in~\eqref{equation: Number of weight 4 codewords in CF dual}. Although the number of blocks provides a tiny piece of information about the incidence structure, it is nevertheless enough to characterize bent and APN functions as we mentioned in Corollaries~\ref{corollary: Coding-theoretic characterization of bent functions} and~\ref{corollary: Coding-theoretic characterization of APN functions}. An essential step further is to analyse, whether one can read off more information about functions, considering the whole incidence structures (but not just the number of blocks).
	
	The motivation of this paper is three-fold. First, following the work~\cite{TangDing2019} of Tang, Ding and Xiong we investigate the coding-theoretic generalizations of vanishing flats, namely the incidence structures $\left(\mathcal{P}\left(\mathcal{C}_{F}\right), \mathcal{B}_{k}\left(\mathcal{C}_{F}\right)\right)$ and $\left(\mathcal{P}\left(\mathcal{C}^{\perp}_{F}\right), \mathcal{B}_{k}\left(\mathcal{C}^{\perp}_{F}\right)\right)$. More precisely, we ask what are the further classes of cryptographically significant $(n,m)$-functions $F$, for which the linear codes $\mathcal{C}_{F}$ and $\mathcal{C}^{\perp}_{F}$ support 2-designs. Second, we introduce a combinatorial generalization of vanishing flats, a partial quadruple system, called nonvanishing flats, based on the modification of the block set~\eqref{equation: Block set of the vanishing flats}. We show that the collection of all nonvanishing flats is an invariant under EA-equivalence and thus can be used as a distinguisher between different classes of functions. Using both vanishing and nonvanishing flats, we provide new characterizations of bent and plateaued $(n,m)$-function and explain the combinatorial difference between these two classes of functions. Finally, we provide coding-theoretic and design-theoretic points of view on the extendability problem of Boolean and vectorial bent functions, which may be considered as a generalization of the bent sum decomposition problem, formulated by Tokareva~\cite[Hypothesis 1]{Tokareva2011}. In particular, we provide a sufficient condition of non-extendability of bent functions in terms of the structural properties of vanishing flats. In this way, the question about the extendability of a given bent functions is reduced to the study of its internal combinatorial information.
	
	The rest of the paper is organized in the following way. In Section~\ref{section 2}, we provide further applications of the extended Assmus-Mattson theorem. First, in Subsection~\ref{subsection 2.1}, we show that the extended Assmus-Mattson theorem, applied to the linear codes of Boolean bent functions may also outperform the transitivity theorem. We show that the automorphism groups of linear codes $\mathcal{C}_f$ and $\mathcal{C}^{\perp}_f$ of Boolean bent functions $f$ on $\F_2^n$ are 2-transitive if and only if $f$ is quadratic. Later on, in Subsection~\ref{subsection: Designs from APN functions}, we show that linear codes of certain APN functions with the classical Walsh spectrum support 2-designs similarly to AB functions; the latter statement was shown in~\cite{TangDing2019}. Moreover, we provide new sufficient conditions for an APN function with the classical Walsh spectrum to be CCZ-inequivalent to a quadratic one. In Section~\ref{section 3}, we consider in details the vanishing flats of Boolean and vectorial bent functions. As it was shown in~\cite[Example 4]{TangDing2019}, linear codes $\mathcal{C}_F$ and $\mathcal{C}^\perp_F$ of all $(n,m)$-bent functions $F$ support 2-designs, from what follows that vanishing flats are 2-designs as well. Using a connection between bent functions and relative difference sets we show that this condition (being a 2-design) is actually sufficient for the perfect nonlinearity. Moreover, we determine the parameters of vanishing flats $\VF(F)$ for $(n,m)$-bent functions $F$. In Section~\ref{section: Nonvanishing flats of plateaued functions}, we generalize the original concept of vanishing flats. In Subsection~\ref{subsection 3.1}, we introduce the notion of nonvanishing flats $\mathcal{NF}_{\vBold}(F)$ for $(n,m)$-functions $F$ and consequently show that the collection of all nonvanishing flats for an $(n,m)$-function $F$ is an invariant under EA-equivalence. In Subsection~\ref{subsection 3.2}, we give a design-theoretic interpretation of the well-known characterization of plateaued functions, given by Carlet~\cite{Carlet2015}. For instance, we show that nonvanishing flats $\mathcal{NF}_{\vBold}(F)$ of $(n,m)$-plateaued functions $F$ are regular incidence structures, i.e., $1$-$(2^n,4,\lambda_\mathbf{v})$ designs (with not necessarily the same $\lambda_\mathbf{v}$). Moreover, we show that the regularity condition is also sufficient for plateauedness, and explain how one can compute the value $\lambda_{\mathbf{v}}$. Further we show that the equiregularity condition, i.e., that nonvanishing flats $\mathcal{NF}_{\vBold}(F)$ are $1$-$(2^n,4,\lambda_\mathbf{v})$ designs with $\lambda_\mathbf{v}=\lambda$ for all $\vBold\in\F_2^m\setminus \{\mathbf{0}\}$, is necessary and sufficient for $s$-plateauedness of $(n,m)$-functions $F$. Finally, in Subsection~\ref{subsection 3.3} we consider nonvanishing flats of bent functions, in particular, we characterize $(n,m)$-bent functions $F$ among the class of plateaued functions as those, for which the nonvanishing flats $\mathcal{NF}_{\vBold}(F)$ are $2$-$(2^n,4,2^{n-m-1})$ designs for all $\vBold\in\F_2^m\setminus \{\mathbf{0}\}$. In Section~\ref{section 5}, we consider extendable and lonely bent functions. In Subsection~\ref{subsection 5.1}, we give a coding-theoretic interpretation of the extendability problem and in Subsection~\ref{subsection 5.2} we provide a design-theoretic framework for studying extendability of bent functions by means of vanishing and nonvanishing flats. In Section~\ref{section 6}, we give concluding remarks and raise some further questions and open problems on $(n,m)$-functions, their linear codes  and incidence structures.
	
	\section{Further applications of the extended Assmus-Mattson theorem}\label{section 2}
	In this section, we investigate further applications of the extended Assmus-Mattson theorem to Boolean and vectorial functions, which, as was shown in~\cite{TangDing2019}, can outperform the original Assmus-Mattson theorem. First we show, that the only Boolean bent functions with 2-transitive automorphism group are quadratic. In this way, for every non-quadratic Boolean bent function $f$ on $\F_2^n$ the linear codes $\mathcal{C}_f$ and $\mathcal{C}_f^\perp$ support 2-designs by the extended Assmus-Mattson theorem, though their automorphism groups are not 2-transitive. We also provide a new application of the extended  Assmus-Mattson theorem, by showing that  linear codes of certain APN functions with the classical Walsh spectrum support 2-designs similarly to AB functions. We also provide new sufficient conditions for an APN function with the classical Walsh spectrum to be CCZ-inequivalent to a quadratic one.
	\subsection{The extended Assmus-Mattson theorem can outperform transitivity theorem}\label{subsection 2.1}
	Using Result~\ref{result: Designs from nice automorphism groups} and the known connections between bent functions, linear codes, addition and translation designs, we prove that vanishing flats of quadratic Boolean bent functions are 2-designs. First we recall that for a subset $A\subseteq G$ of a finite group $(G,+)$ the \emph{development} $\dev(A)$ of $A$ is an incidence structure, whose points are the elements in $G$, and whose blocks are the translates $A + g := \{a+g:a\in A \}$, where $g\in G$.
	
	\begin{definition}
		For a Boolean function $f$ on $\F_2^n$, we define the following two incidence structures:
		\begin{itemize}
			\item $\dev(\support{f})$, where $\support{f}:=\{ \xBold \in \F_2^n \colon f(\xBold)=1 \}$;
			\item $\D(f)=(\mathcal{P}, \mathcal{B})$ with $\mathcal{P}:=\{ \xBold : \xBold\in\F_2^n \}$ and $\mathcal{B}:=\{ \support{f\oplus l} \colon l\in\mathcal{A}_n \mbox{ with } \wt(f\oplus l)=\wt(f) \}$.
		\end{itemize}		
	\end{definition}
	
	For an arbitrary Boolean function $f$ on $\F_2^n$ the defined incidence structures $\dev(\support{f})$ and $\D(f)$ are not necessarily 2-designs. However, in the case if a function $f$ is bent, both of them are 2-designs with the classical Hadamard parameters. For this reason, we will call the incidence structures $\dev(\support{f})$ and $\D(f)$ of a bent function $f$ on $\F_2^n$  the \emph{translation} and \emph{addition} designs, respectively.

	\begin{result}\label{result: Bent functions and designs}
		Let $f$ be a Boolean function on $\F_2^n$. The following statements are equivalent.
		\begin{enumerate}
			\item The function $f$ is bent.
			\item \cite[Chapter 3]{lander_1983} The translation design $\dev(\support{f})$ is a $2$-$(2^{n},2^{n-1}\pm 2^{n/2-1},2^{n-2} \pm 2^{n/2-1})$ design with $+$ if $f(\mathbf{0})=1$ and $-$ otherwise.
			\item \cite[Chapter 9]{Bending1993} The addition design $\D(f)$ is equal to $(\mathcal{P}(\mathcal{C}_f),\mathcal{B}_{2^{n-1}+2^{n/2-1}}(\mathcal{C}_f))$ if $f(\mathbf{0})=1$, otherwise it is equal to $(\mathcal{P}(\mathcal{C}_f),\mathcal{B}_{2^{n-1}-2^{n/2-1}}(\mathcal{C}_f))$. In turn, $\D(f)$ is a $2$-$(2^{n},2^{n-1}\pm 2^{n/2-1},2^{n-2} \pm 2^{n/2-1})$ design, respectively.
		\end{enumerate}
	\end{result}

	In the following, we use Result~\ref{result: Designs from nice automorphism groups} and connections between translation and addition designs, in order to characterize quadratic Boolean bent functions in terms of the 2-transitivity of the automorphism group.
	
	\begin{theorem}\label{proposition: Aut group of the quadratic bent function}
		Let $f$ be a Boolean bent function on $\F_2^n$. The automorphism groups of linear codes $\mathcal{C}_f$ and $\mathcal{C}^\perp_f$ are 2-transitive if and only if $f$ is quadratic. Consequently, the linear codes $\mathcal{C}_f$ and $\mathcal{C}^\perp_f$ of quadratic Boolean bent functions $f$ on $\F_2^n$ support 2-designs.
	\end{theorem}
	\begin{proof}
		Let $f$ be a quadratic Boolean bent function on $\F_2^n$. Without loss of generality we assume that $f(0)=0$, in this way, $\dev(\support{f})$ is a $2$-$(2^{n},2^{n-1}- 2^{n/2-1},2^{n-2} - 2^{n/2-1})$  design by Result~\ref{result: Bent functions and designs}. By~\cite[Lemma 2.3]{Cusick2011} any two Boolean functions $g$ and $g'$ on $\F_2^n$ are affine equivalent (i.e., there exists an affine permutation $A$ of $\F_{2}^{n}$ such that $g=g'\circ A$) if and only if $\#\support{g}=\#\support{g'}$. Consequently, for a quadratic Boolean bent function $f$ on $\F_2^n$ with $f(0)=0$ its translation design $\dev(\support{f})$ is isomorphic to the symplectic $2$-$(2^{n},2^{n-1}- 2^{n/2-1},2^{n-2} - 2^{n/2-1})$ design $S$, which is realized as the translation design of the ``dot product'' quadratic bent function on $\F_2^n$, see~\cite[Chapter 3]{lander_1983}. The automorphism group of $S$, being a semidirect product of the translation group $\Sigma$ of the affine space $\mbox{AG}(n,2)$ with the symplectic group $\mbox{Sp}(n,2)$, is 2-transitive, as it was shown by Kantor~\cite{KANTOR197543}. Moreover, the symplectic $2$-$(2^{n},2^{n-1}- 2^{n/2-1},2^{n-2} - 2^{n/2-1})$ design $S$ with 2-transitive automorphism group is unique up to isomorphism for any even $n$, as it was shown by Kantor~\cite{Kantor1985}. Further, for a quadratic bent function $f$ on $\F_2^n$ its translation $\dev(\support{f})$ and addition $\D(f)$ designs are isomorphic as it was shown by Bending~\cite[Theorem 11.9]{Bending1993}. The characterization of quadratic bent functions now follows from the fact that EA-equivalence of Boolean bent functions $f$ and $f'$ on $\F_2^n$ and isomorphism of the addition designs $\D(f)$ and $\D(f')$ are the same concepts, see~\cite{DillonSchatz1987} and~\cite[Corollary 10.6]{Bending1993}. Finally, since the incidence matrix $\mathbf{M}(\D(f))$ is a generator matrix of the linear code $\mathcal{C}_f$, we have that $\Aut(\mathcal{C}_f)$ as well as $\Aut(\mathcal{C}_f^\perp)$ for a Boolean bent function $f$ on $\F_2^n$ are 2-transitive if and only $f$ is quadratic. In this way, by Result~\ref{result: Designs from nice automorphism groups}, the linear codes $\mathcal{C}_f$ and $\mathcal{C}_f^\perp$ of quadratic Boolean bent functions $f$ on $\F_2^n$  support 2-designs. In particular, the incidence structure supported by codewords of weight 4 in $\mathcal{C}^\perp_f$, i.e., the vanishing flats $\VF(f)$, is a 2-design.
	\end{proof}
	\begin{corollary}
		Let $f$ be a quadratic bent function on $\F_2^n$ with $n=2k$. Then the order of the automorphism group of the function $f$ is given by
		\begin{equation}\label{equation: Aut group of the quadratic bent function}
		|\Aut(f)|=2^{n} \cdot 2^{k^2} \prod_{i=1}^k (2^{2i} - 1).
		\end{equation}
		Here by $\Aut(f)$ we mean one of the following groups: $\Aut(\D(f))$, $\Aut(\VF(f))$, $\Aut(\mathcal{C}_f)$ and $\Aut(\mathcal{C}_f^\perp)$.
	\end{corollary}
	
	\begin{remark}\label{remark: Vanishing flaats of (6,m)-bent functions}
		In general, an $(n,m)$-bent function does not necessarily need to have a 2-transitive automorphism group. Using a Magma program~\cite{MR1484478} and the list of representatives of the equivalence classes (denoted by $F^m_i$) of $(6,m)$-bent functions, obtained in~\cite{polujan2020design}, it is possible to check that, up to EA-equivalence, the only $(6,m)$-bent functions, which have a 2-transitive automorphism group, are the following quadratic functions: $F^1_1,F^2_1,F^3_1$ and $F^3_3$. Since the automorphism group of the quadratic function $F^3_2$ is not 2-transitive, we conclude that the property of an $(n,m)$-bent function to be quadratic does not in general imply the 2-transitivity of its automorphism group.
	\end{remark}
	
	In this way, we conclude that 2-transitivity of the automorphism group is a characteristic property of quadratic Boolean bent functions. However, despite the fact that all Boolean bent functions but quadratic on $\F_2^n$  and many vectorial bent functions in a small number of variables do not have a 2-transitive automorphism group, their linear codes (and in general of all $(n,m)$-bent functions) still support 2-designs due to the extended Assmus-Mattson theorem, see ~\cite[Example 4]{TangDing2019}. In the following subsection, we provide another important class of vectorial Boolean functions, whose linear codes support 2-designs.
	\subsection{Designs from APN functions with the classical Walsh spectrum}\label{subsection: Designs from APN functions}
	Recently, Tang, Ding and Xiong~\cite{TangDing2019} proved that not only the vanishing flats $\VF(F)$ of differentially two-valued $s$-plateaued $(n,n)$-functions $F$ are 2-designs, but that the linear codes $\mathcal{C}_F$ and $\mathcal{C}_F^\perp$ support 2-designs.

\begin{result}\cite[Theorem 6.4]{TangDing2019}\label{result: 2-designs from s-plateaued functions}
	Let $F$ be a differentially two-valued $s$-plateaued $(n,n)$-function. Then the code $\mathcal{C}_F$ and its dual $\mathcal{C}_F^{\perp}$ support 2-designs.
\end{result}
\noindent As a corollary, this result implies that the linear codes $\mathcal{C}_F$ and $\mathcal{C}^\perp_F$ of all AB functions $F$ on $\F_2^n$, which belong to the class of APN functions with the classical Walsh spectrum, support 2-designs. In the following statement, we provide another subclass of APN functions $F$ on $\F_2^n$ with the classical Walsh spectrum, whose linear codes $\mathcal{C}_F$ and $\mathcal{C}^\perp_F$ support 2-designs.
\begin{definition}
	Let $F$ be an APN function on $\F_2^n$. Without loss of generality we assume that $F(0)=0$. We say that the function $F$ has \emph{the classical Walsh spectrum}, if it is either \emph{AB}, i.e., $n$ is odd and 
	\begin{equation}
	\Lambda_{F}=\left\{ *  2^{n}[1], \  0\left[\left(2^{n-1}+1\right)\left(2^{n}-1\right)\right], \ \pm 2^{(n+1) / 2}\left[\left(2^{n}-1\right)\left(2^{n-2} \pm 2^{(n-3) / 2}\right)\right] *\right\},
	\end{equation}
	or $n$ is even and $F$ has the Walsh spectrum of the Gold APN functions $f\colon x\in\F_{2^n}\mapsto x^{2^i+1}$ with $\gcd(i,n)=1$, namely,
	\begin{equation}
	\begin{array}{c}
	\Lambda_{F}=\left\{*  2^{n}[1], \ 0\left[\left(2^{n}-1\right)\left(2^{n-2}+1\right)\right], \ \pm 2^{(n+2) / 2}\left[\frac{1}{3}\left(2^{n}-1\right)\left(2^{n-3} \pm 2^{(n-4)/2}\right)\right],\right. \\
	\left.\pm 2^{n / 2}\left[\frac{2}{3}\left(2^{n}-1\right)\left(2^{n-1} \pm 2^{(n-2) / 2}\right)\right] *\right\}.
	\end{array}
	\end{equation}
\end{definition}
Recall that the \emph{dual} of a Boolean bent function $f\colon \F_2^n\rightarrow\F_2$ is a bent function $\tilde{f}\colon\F_2^n\rightarrow\F_2$, defined by $\hat{\chi}_f(\aBold)=2^{n/2}(-1)^{\tilde{f}(\aBold)}$. Bending in his thesis~\cite{Bending1993} showed, how one can construct an incidence matrix of the addition design $\D(f)$ of a bent function $f$ on $\F_2^n$ with the help of the dual function $\tilde{f}$ of $\F_2^n$.
\begin{result}\label{result: Bendings designs}
	\cite[Theorem 9.6]{Bending1993} Let $f$ be a bent function on $\F_2^n$. An incidence matrix $\mathbf{M}(\D(f))$ of the addition design $\D(f)$ can be constructed in the following way
	\begin{equation}\label{equation: Incidence matrix of the addition design}
	\mathbf{M}(\D(f))= (m_{\xBold,\yBold})_{\xBold,\yBold\in\F_2^n},\mbox{ where } m_{\xBold,\yBold}=\tilde{f}(\xBold)\oplus f(\yBold)\oplus \langle\xBold,\yBold\rangle_n\oplus\tilde{f}(\mathbf{0}).
	\end{equation}
\end{result}
With the use of this result and the extended Assmus-Mattson theorem we proceed with the main result of the section.
\begin{theorem}\label{theorem: APNs with classical Walsh spectrum support 2-designs}
	Let $F$ be an APN function on $\F_2^n$ with $n=2k$, which has the classical Walsh spectrum. If the function $F$ is CCZ-equivalent to a function $F'$ on $\F_2^n$ having only bent and semi-bent nonzero components, then the linear codes $\mathcal{C}_F$ and $\mathcal{C}^{\perp}_F$ support $2$-designs.
\end{theorem}
\begin{proof}
	We show that the conditions of the extended Assmus-Mattson theorem (Result~\ref{result: Extended Assmus-Mattson theorem}) are fulfilled for
	linear codes $\mathcal{C}_F$ and $\mathcal{C}^\perp_{F'}$ with integers $t=2$ and $s:=|S|=3$, where the set $S$ is defined as follows $S:=\{ 2^{n-1}, 2^{n-1}\pm 2^{n/2} \}$. Since nonzero component functions of $F'$ are bent and semi-bent, we have that for any codeword $\cBold\in\mathcal{C}_{F'}\setminus\{\oBold,\mathbf{j}_{2^n} \}$ (where $\mathbf{j}_{2^n}$ denotes the all-one-vector of length $2^n$) holds
	\begin{equation*}
		\wt(\cBold)\in W:= \{ 2^{n-1}\pm 2^{n/2-1}, 2^{n-1}, 2^{n-1}\pm 2^{n/2} \},
	\end{equation*}
	since possible Hamming weights of bent functions are $2^{n-1}\pm 2^{n/2-1}$ and possible Hamming weights of semi-bent functions are $2^{n-1}$ and $2^{n-1}\pm 2^{n/2}$, respectively. Note that the number of bent components of $F'$ is equal to $\frac{2}{3}(2^n-1)$ and the number of semi-bent components is equal to $\frac{1}{3}(2^n-1)$.
	By Result~\ref{result: Bendings designs}, from any component function $F'_{\bBold}$ of $F'$, which is bent, one can construct 2-designs $(\mathcal{P}(\mathcal{C}_{F'_{\bBold}}),\mathcal{B}_{\ell}(\mathcal{C}_{F'_{\bBold}}))$ with $\ell\in W\setminus S= \{ 2^{n-1}\pm 2^{n/2-1} \}$ using the dual function $\tilde{F'}_{\bBold}$ as described in Result~\ref{result: Bendings designs}. Clearly, any two component bent functions $F'_{\bBold}$ and $F'_{\bBold'}$ with $\bBold\neq \bBold'$ do not differ by an affine function. In this way, we have that for $\ell\in \{ 2^{n-1}\pm 2^{n/2-1} \}$ the incidence structures $(\mathcal{P}(\mathcal{C}_{F'}),\mathcal{B}_{\ell}(\mathcal{C}_{F'}))$ are $2$-$(2^{n},2^{n-1}\pm 2^{n/2-1},\frac{2}{3}(2^n-1)\cdot(2^{n-2} \pm 2^{n/2-1}))$ designs, since they are obtained by a disjoint union of $2$-$(2^{n},2^{n-1}\pm 2^{n/2-1},$ $2^{n-2} \pm$ $2^{n/2-1})$ designs  $(\mathcal{P}(\mathcal{C}_{F'_{\bBold}}),\mathcal{B}_{w}(\mathcal{C}_{F'_{\bBold}}))$ having no repeated blocks. Since the function $F'$ is APN, the minimum distance $d^{\perp}$ of $\mathcal{C}^{\perp}_{F'}$ is equal to $d^{\perp}=6$. In this way, for any $0\le \ell^\perp \le s+t-1=4$ we have $\mathcal{B}_{\ell^\perp}(\mathcal{C}^\perp_{F'})=\varnothing$ and thus $(\mathcal{P}(\mathcal{C}^\perp_{F'}),\mathcal{B}_{\ell^\perp}(\mathcal{C}^\perp_{F'}))$ are trivial 2-designs. Since for all $\ell \in W \setminus S$ and $\ell^\perp$ with $0 \leq \ell^\perp \leq 4$, the incidence structures $\left(\mathcal{P}(\mathcal{C}_{F'}), \mathcal{B}_{\ell}(\mathcal{C}_{F'})\right)$ and
	$\left(\mathcal{P}\left(\mathcal{C}_{F'}^{\perp}\right), \mathcal{B}_{\ell^{\perp}}\left(\mathcal{C}_{F'}^{\perp}\right)\right)$ are $2$-designs, we have that the linear codes $\mathcal{C}_{F'}$ and $\mathcal{C}^{\perp}_{F'}$ support $2$-designs.
\end{proof} 

\begin{remark}
	In general, Theorem~\ref{theorem: APNs with classical Walsh spectrum support 2-designs} does not hold for all APN functions, as the following well-known example of an APN function with a nonclassical Walsh spectrum shows. Let $F$ be the following APN function on $\F_{2^6}$ from Dillon's list~\cite{browning2009apn}
	\begin{equation}
	F(x)=x^3 + a^{11} x^5 + a^{13}x^9 + x^{17} + a^{11} x^{33} + x^{48},
	\end{equation}
	where $a$ is a primitive element of $\F_{2^6}$, satisfying $a^6 + a^4 + a^3 + a + 1=0$. Using a Magma~\cite{MR1484478} program it is not difficult to check, that the Walsh spectrum of $f$, given below,
	\begin{equation*}
	\Lambda_{F}:=\{* -32[1], \ -16[96], \ -8[1288], \ 0[891], \ 8[1656], \ 16[160], \ 32[3] *\}
	\end{equation*}
	is nonclassical and the incidence structures supported by the codewords of the minimum weight are 1-designs, but not 2-designs. For instance:
	\begin{itemize}
		\item $(\mathcal{P}(\mathcal{C}_F),\mathcal{B}_{16}(\mathcal{C}_F))$ is a $1$-$(64, 16, 1)$ design with 4 blocks;
		\item $(\mathcal{P}(\mathcal{C}_F^\perp),\mathcal{B}_{6}(\mathcal{C}_F^\perp))$ is a $1$-$(64, 6, 1986)$ design with 21184 blocks.
	\end{itemize}
\end{remark}
Now, we describe a big class of APN functions with the classical Walsh spectrum, satisfying the conditions of Theorem~\ref{theorem: APNs with classical Walsh spectrum support 2-designs}. Note that this class contains most of the known examples and constructions of APN functions.
\begin{theorem}\label{theorem: Quadratic APNs with classical Walsh spectrum support 2-designs}
	Let $F$ be an APN function on $\F_2^n$ with $n=2k$, which has the classical Walsh spectrum. If the function $F$ is CCZ-equivalent to a quadratic APN function, then the linear codes $\mathcal{C}_F$ and $\mathcal{C}^{\perp}_F$ support $2$-designs.
\end{theorem}
\begin{proof}
	Let $F'$ be a quadratic APN function on $\F_2^n$ such that $F$ and $F'$ are CCZ-equivalent. From the fact that the extended Walsh spectrum is invariant under CCZ-equivalence, we have that $|\hat{\chi}_{F'}(\aBold,\bBold)|\in\left\{ 0,2^{\frac{n}{2}},2^{\frac{n+2}{2}} \right\}$ for all $\aBold\in\F_2^n$, $\bBold\in\F_2^m\setminus\{\mathbf{0}\}$. Since the function $F'$ is quadratic, it is plateaued, and thus for any component function $F'_\bBold$ with $\bBold\in\F_2^m\setminus\{ \mathbf{0} \}$, we have $|\hat{\chi}_{F'_\bBold}(\aBold)|\in\left\{ 2^{\frac{n}{2}} \right\}$ for all $\aBold\in\F_2^n$ if and only if $F'_\bBold$ is bent, and  $|\hat{\chi}_{F'_\bBold}(\aBold)|\in\left\{ 0,2^{\frac{n+2}{2}} \right\}$ for all $\aBold\in\F_2^n$ if and only if $F'_\bBold$  is semi-bent. In this way, nonzero components of $F'$ are bent and semi-bent. By Result~\ref{theorem: Support designsare  CCZ-invariant} and Theorem \ref{theorem: APNs with classical Walsh spectrum support 2-designs}, the linear codes $\mathcal{C}_F$ and $\mathcal{C}^{\perp}_F$ support $2$-designs, since functions $F$ and $F'$ are CCZ-equivalent.
\end{proof}
With Theorem~\ref{theorem: Quadratic APNs with classical Walsh spectrum support 2-designs}, we derive the following sufficient conditions for an APN function with the classical Walsh spectrum to be CCZ-inequivalent to a quadratic function.

\begin{corollary}\label{theorem: Sufficient conditions for CCZ-ineq to a qaudratic}
	Let $F$ be an APN function on $\F_2^n$ with $n=2k$, which has the classical Walsh spectrum.
	\begin{enumerate}
		\item If there exists an integer $\ell$, satisfying $1<\ell<2^n$ such that the incidence structure $\left(\mathcal{P}(\mathcal{C}_F), \mathcal{B}_{\ell}(\mathcal{C}_F)\right)$ is not a 2-design, then the APN function $F$ is CCZ-inequivalent to a quadratic function.
		\item If there exists an integer $\ell^\perp$, satisfying $6<\ell^\perp<2^n$ such that the incidence structure $\left(\mathcal{P}(\mathcal{C}^\perp_F), \mathcal{B}_{\ell^\perp}(\mathcal{C}^\perp_F)\right)$ is not a 2-design, then the APN function $F$ is CCZ-inequivalent to a quadratic function.
	\end{enumerate}
	\end{corollary}

\begin{example}
	Edel and Pott~\cite{Edel_non_quadratic_APN:200959} showed that the following APN function $F$ on $\F_{2^6}$ given by 
	\begin{equation}\label{equation: Edel-Pott APN} 
		\begin{split}
			F(x) & = x^3 + a^{17}(x^{17} + x^{18} + x^{20} + x^{24}) + a^{14}(a^{18}x^9 + a^{36} x^{18} + a^9 x^{36} + x^{21} + x^{42} \\
			&+Tr(a^{27} x + a^{52} x^3 + a^6 x^5 + a^{19} x^7 + a^{28} x^{11} + a^2 x^{13})),
		\end{split}
	\end{equation}
	where $a$ is a primitive element of $\F_{2^6}$, satisfying $a^6 + a^4 + a^3 + a + 1=0$ and $Tr(x)$ denotes the absolute trace of $x\in\F_{2^6}$, is CCZ-inequivalent to a quadratic APN function using the information about the automorphism group of the incidence structure $\dev(\graph{F})$. Now we show that the function $F$ is CCZ-inequivalent to a quadratic one using Corollary~\ref{theorem: Sufficient conditions for CCZ-ineq to a qaudratic}. First, we observe that the function $F$ has the classical Walsh spectrum. The weight distribution of the linear code $\mathcal{C}_F$ is given by 
	\begin{equation*}
		W_{\mathcal{C}_F}(z)=1+336 z^{24}+2688 z^{28} + 2142 z^{32}+ 2688 z^{36} + 336 z^{40}+z^{64}.
		%[ <0, 1>, <24, 336>, <28, 2688>, <32, 2142>, <36, 2688>, <40, 336>, <64, 1> ]
	\end{equation*}
	With {Magma}~\cite{MR1484478}, it is possible to check that the incidence structure $\left(\mathcal{P}\left(\mathcal{C}_F\right), \mathcal{B}_{24}\left(\mathcal{C}_F\right)\right)$ supported by the codewords of the minimum weight of $\mathcal{C}_F$ is a 1-(64, 24, 126) design with 336 blocks, but not a 2-design. By Corollary~\ref{theorem: Sufficient conditions for CCZ-ineq to a qaudratic}, the function $F$ defined in~\eqref{equation: Edel-Pott APN}  is CCZ-inequivalent to a quadratic function.
\end{example}
	For more details on (vectorial) Boolean functions, their codes and incidence structures, we refer to~\cite[Chapter 3]{Polujan_PhD}.
	\section{Vanishing flats of Boolean and vectorial bent functions}\label{section 3}
Vanishing flats $\VF(F)$ of $(n,m)$-bent functions $F$, being supports of codewords of weight 4 in the linear code $\mathcal{C}^\perp_F$, are 2-designs, as it was observed in~\cite[Example 4]{TangDing2019}. In this section, we explain the combinatorial structure of vanishing flats of bent functions and, consequently, compute the parameters of these designs, i.e., we show that vanishing flats $\VF(F)$ of $(n,m)$-bent functions $F$ are $2$-$(2^n,4,2^{n-m-1}-1)$ designs. Moreover, we show that this design-theoretic condition is also sufficient for the perfect nonlinearity. The key ingredient of the proof is the characterization of $(n,m)$-bent functions in terms of relative difference sets.
\begin{definition}
	Let $(G,+)$ be a finite group of order $\mu \cdot \nu$ and $N$ be a normal subgroup of order $\nu$ of $G$. A subset $R \subseteq G$ of $G$ is called a \emph{relative $(\mu, \nu , k, \lambda)$-difference set} relative to the subgroup $N$ if $|R| = k$ and the list of differences $r-r'$ with $r, r' \in R$ contains all the elements of $G\setminus N$ exactly $\lambda$ times. Moreover, no nonzero element in $N$, the so-called \emph{forbidden subgroup}, occurs in this list of differences.
\end{definition}

\begin{result}\cite[Theorem 1]{Pott04}\label{result: Vectorial bent functions and RDS} Let $n$ be even. The following statements are equivalent.
	\begin{enumerate}
		\item An $(n,m)$-function $F$ is bent.
		\item The graph $\graph{F}\subseteq\F_2^{n}\times \F_2^{m}$ is a relative $\left(2^{n}, 2^{m}, 2^{n}, 2^{n-m}\right)$-difference set in $\F_2^{n}\times \F_2^{m}$ relative to the subgroup $N=\{(\mathbf{0}, \mathbf{y})\colon \mathbf{y}\in\F_2^m \}$.
	\end{enumerate}
\end{result}

With the help of this connection, we characterize $(n,m)$-bent functions in terms of vanishing flats having the property to be 2-designs.

\begin{theorem}\label{theorem: Bent functions via vanishing flats}
	Let $F$ be an $(n,m)$-function. The following statements are equivalent.
	\begin{enumerate}
		\item $F$ is an $(n,m)$-bent function.
		\item $\mathcal{VF}(F)$ is a $2$-$(2^n,4,2^{n-m-1}-1)$ design.
	\end{enumerate}
\end{theorem}
\begin{proof}
	\emph{1.$\Rightarrow$2.} Let $F$ be an $(n,m)$-bent function and let $\VF(F)=(\mathcal{P},\mathcal{B})$ be the vanishing flats of $F$. We will show that any two different points $\xBold_1,\xBold_2$ of $\mathcal{P}$ are contained in exactly $2^{n-m-1}-1$ blocks of $\mathcal{B}$. We define $\aBold:=\xBold_1\oplus\xBold_2$ and let $\vBold:=F(\xBold_1)\oplus F(\xBold_2)$. Since the graph $\graph{F}$ is a $\left(2^{n}, 2^{m}, 2^{n}, 2^{n-m}\right)$-difference set in the group $G=\F_2^{n} \times \F_2^{m}$ relative to the forbidden subgroup $N=\{(\mathbf{0}, \mathbf{y})\colon \mathbf{y}\in\F_2^m \}$, the element  $g:=\graphElt{\aBold}{\vBold} \in G\setminus N$ has $2^{n-m}-2$ further representations
	\begin{equation}\label{equation: Combiatorial view point on vanishing flats}
		\graphElt{\xBold_1}{F(\xBold_1)}\oplus \graphElt{\xBold_2}{F(\xBold_2)}=g=\graphElt{\xBold_3}{F(\xBold_3)}\oplus \graphElt{\xBold_4}{F(\xBold_4)}
	\end{equation}
	with $\{ \xBold_3,\xBold_4 \} \neq \{ \xBold_1,\xBold_2 \}$. In this way, any 2-subset $ \{ \xBold_1,\xBold_2 \}$ is contained in exactly $2^{n-m-1}-1$ blocks $\{ \xBold_1,\xBold_2,\xBold_3,\xBold_4 \}$, satisfying
	\begin{equation}\label{equation: Property of vanishing flats}
		\graphElt{\xBold_1}{F(\xBold_1)}\oplus \graphElt{\xBold_2}{F(\xBold_2)}\oplus \graphElt{\xBold_3}{F(\xBold_3)}\oplus \graphElt{\xBold_4}{F(\xBold_4)}=\graphElt{\mathbf{0}}{\mathbf{0}},
	\end{equation}
	from what follows that $\mathcal{VF}(F)$ is a $2$-$(2^n,4,2^{n-m-1}-1)$ design.
	
	\noindent\emph{2.$\Rightarrow$1.} The statement follows from the fact, that the obtained number of blocks~\eqref{equation: The number of vanishing flats of an (n,m)-bent function} of a $2$-$(2^n,4,2^{n-m-1}-1)$ design $\mathcal{VF}(F)$ corresponds to the minimum possible value \eqref{equation: Bentness via CF dual} of the weight 4 codewords $A_4$ in the linear code $\mathcal{C}^\perp_F$, which is attained if and only if $F$ is an $(n,m)$-bent function.
\end{proof}

\begin{remark}\label{remark: Number of vanishing flats of a bent function}
	The proof of the previous statement gives a constructive combinatorial way to determine the number
	of vanishing flats $|\mathcal{VF}_{F}|$ of $F$ as follows:
	\begin{itemize}
		\item there are $2^{n+m}-2^{m}$ ways to pick an element $g\in G\setminus N$;
		\item for a selected element $g$, there exist $2^{n-m}$ ways to choose the left-hand side in~\eqref{equation: Combiatorial view point on vanishing flats}, what gives $2^{n-m}-2$ remaining choices of the right-hand side;
		\item in this way, we have $\left(2^{n+m}-2^m\right) \cdot 2^{n-m} \cdot \left(2^{n-m}-2\right)$ ``ordered'' vanishing flats of $F$, i.e., quadruples $(\xBold_1,\xBold_2,\xBold_3,\xBold_4)$ with the property~\eqref{equation: Property of vanishing flats} . 
	\end{itemize}
	Dividing this number by 24, we get that the number of blocks of $\VF(F)$ is equal to
	\begin{equation}\label{equation: The number of vanishing flats of an (n,m)-bent function}
		|\mathcal{VF}_{F}|=\frac{\left(2^{n+m}-2^m\right) \cdot 2^{n-m} \cdot \left(2^{n-m}-2\right)}{24},
	\end{equation}
	which after expanding and simplifying coincides with the value in \eqref{equation: Bentness via CF dual}.
\end{remark}
As it was mentioned in Result~\ref{theorem: Support designsare  CCZ-invariant}, vanishing flats are invariants under CCZ-equivalence for $(n,m)$-functions. Conversely, two $(n,m)$-functions $F,F'$ for which the vanishing flats $\VF(F)$ and $\VF(F')$ are isomorphic, are not necessarily CCZ-equivalent, as the case of APN functions shows, since the obtained incidence structures are trivial. However, according to our computational results, the converse is also true for $(6,m)$-bent functions.
\begin{theorem}\label{theorem: EA-equivalence via vanishing flats}
	Let $F$ and $F'$ be two $(6,m)$-bent functions with $m\ge2$. The following statements are equivalent.
	\begin{enumerate}
		\item Bent functions $F$ and $F'$ are extended-affine equivalent.
		\item Vanishing flats $\VF(F)$ and $\VF(F')$ are isomorphic.
	\end{enumerate}
	Moreover for any $(6,m)$-bent function $F$ the linear code $\mathcal{C}_F^\perp$ is spanned by the codewords of minimum weight.
\end{theorem}
\begin{proof}
	All computations about equivalence and isomorphism are carried out with {Magma}~\cite{MR1484478} using the representatives of the equivalence classes of $(6,m)$-bent functions from~\cite{polujan2020design}.
\end{proof}
\section{Nonvanishing flats of plateaued functions}\label{section: Nonvanishing flats of plateaued functions}
In this section, we introduce a combinatorial generalization of vanishing flats by modifying the definition of the block set in~\eqref{equation: Block set of the vanishing flats} and consequently use this generalization in order to derive new characterizations of plateaued and bent functions.
\subsection{Definition and invariance under EA-equivalence}\label{subsection 3.1}
First, we give a formal definition of nonvanishing flats and show that the collection of nonvanishing flats of an $(n,m)$-function is an invariant under EA-equivalence.
\begin{definition}
	Let $F$ be an $(n,m)$-function. We define a partial quadruple system, called the \emph{nonvanishing flats of the $(n,m)$-function $F$ with respect to the nonzero vector $\vBold\in\F_2^m$}, as the incidence structure $\NF_{\vBold}(F):=(\mathcal{P}, \NF_{\vBold,F})$ where the point set is given by $\mathcal{P}=\{ \xBold \colon \xBold \in\F_2^n \}$ and the block set $\NF_{\vBold,F}$ is defined as follows
	\begin{equation}\label{equation: Block set of the nonvanishing flats}
		\NF_{\vBold,F}= \left\{ \{ \xBold_1,\xBold_2,\xBold_3,\xBold_4  \} \colon \bigoplus\limits_{i=1}^{4} \graphElt{\xBold_i}{F(\xBold_i)}=\graphElt{\mathbf{0}}{\mathbf{\vBold}} \mbox{ for } \xBold_i\in\F_2^n  \right\}.
	\end{equation}
\end{definition}

\begin{remark}
	Clearly, for an arbitrary $(n,m)$-function $F$ the collection of incidence structures $\{ \VF(F) \}\cup \{ \NF_{\vBold}(F) \colon \vBold\in\F_2^m\setminus \{ \mathbf{0} \} \}$ forms a partition of \emph{the affine Steiner quadruple system} $SQS(2^n):=(\mathcal{P},\mathcal{B})$ with point and block sets being defined as follows
	\begin{equation*}
		\mathcal{P}=\{ \xBold \colon \xBold \in \F_2^n \} \quad \mbox{and} \quad \mathcal{B}=\{ \{ \xBold_1,\xBold_2,\xBold_3,\xBold_4 \}\colon \xBold_1 \oplus\xBold_2\oplus\xBold_3\oplus\xBold_4=\mathbf{0} \mbox{ for } \xBold_i\in\F_2^n   \},
	\end{equation*}
	which is a $3$-$(2^n,4,1)$ design.
\end{remark}

As we mentioned in Result~\ref{theorem: Support designsare  CCZ-invariant}, the vanishing flats $\VF(F)$ are invariants under CCZ-equivalence for $(n,m)$-functions. Further we will show that the collection of all nonvanishing flats $\{\NF_{\vBold}(F) \colon \vBold\in\F_2^m\setminus\{ \mathbf{0}\} \}$ is an invariant under EA-equivalence for $(n,m)$-functions. First, we recall the following characterization of EA-equivalence for perfect nonlinear functions~\cite[Section 5]{EdelP09}.

\begin{result}\cite{EdelP09}
	Let $F$ and $F'$ be two $(n,m)$-functions. Then $F$ and $F'$ are extended-affine equivalent if and only if there exists an affine permutation $\mathcal{L}$ of $ \ \F_2^n\times \F_2^m$ of the form
	\begin{equation}\label{equation: EA permutation}
		\mathcal{L}\colon\graphElt{\xBold}{\yBold} \mapsto \begin{pmatrix}
			\mathbf{A}_{11} & \mathbf{O}_{\phantom{12}}\\
			\mathbf{A}_{21} & \mathbf{A}_{22}\end{pmatrix} \graphElt{\xBold}{\yBold} \oplus \graphElt{\aBold}{\bBold}.
	\end{equation}
	such that $\mathcal{L}(\graph{F})=\graph{F'}$.
\end{result}
Following the original proof of the invariance of vanishing flats under CCZ-equivalence~\cite[Theorem II.1.]{Li20VanishingFlats} and using the mentioned characterization of EA-equivalence, we proof the following result.
\begin{theorem}
	\label{thm-CCZinvariant}
	Let $F$ and $F'$ be two EA-equivalent $(n,m)$-functions and let $\mathcal{L}$ be an affine permutation of $\F_2^n\times \F_2^m$ of the form~\eqref{equation: EA permutation} such that $\mathcal{L}(\graph{F})=\graph{F'}$. Then for any nonzero $\vBold\in\F_2^m$ the block $\{\xBold_1,\xBold_2,\xBold_3,\xBold_4\} \in \NF_{\vBold}(F)$ if and only if $\{\xBold'_1,\xBold'_2,\xBold'_3,\xBold'_4\} \in \NF_{\vBold'}(F')$, where $\xBold'_i=\mathbf{A}_{11} \xBold_i \oplus \cBold$ and $\vBold'=\mathbf{A}_{22}\vBold$.
\end{theorem}
\begin{proof}
	Let $\{\xBold_1,\xBold_2,\xBold_3,\xBold_4\}$ be a block of $\NF_{\vBold}(F)$. Then the following holds,
	\begin{equation*}
		\graphElt{\xBold_1}{F(\xBold_1)}\oplus \graphElt{\xBold_2}{F(\xBold_2)}\oplus \graphElt{\xBold_3}{F(\xBold_3)}\oplus \graphElt{\xBold_4}{F(\xBold_4)}=\graphElt{\mathbf{0}}{\mathbf{v}}.
	\end{equation*}
	Let $\pi$ be a mapping from $\F_2^n$ to $\F_2^n$, such that $\pi(\xBold)=\mathbf{A}_{11}\xBold\oplus \cBold$. By definition of EA-equivalence, $\pi$ induces a permutation on $\F_2^n$. Further, let $\xBold'_i:=\pi(\xBold_i)$ and $F'(\xBold'_i)=\mathbf{A}_{21} \xBold_i \oplus \mathbf{A}_{22} F(\xBold_i) \oplus \bBold$. We then have
	\begin{equation*}
		\graphElt{\xBold'_1}{F'(\xBold'_1)}\oplus \graphElt{\xBold'_2}{F'(\xBold'_2)}\oplus \graphElt{\xBold'_3}{F'(\xBold'_3)}\oplus \graphElt{\xBold'_4}{F'(\xBold'_4)}=\graphElt{\mathbf{0}}{\mathbf{A}_{22}\bigoplus\limits_{i=1}^4F(\xBold_i) }=\graphElt{\mathbf{0}}{\mathbf{v'}},
	\end{equation*}
	where $\vBold':=\mathbf{A}_{22} \vBold$. In this way, $\{\xBold'_1,\xBold'_2,\xBold'_3,\xBold'_4\}$ is a block of $\NF_{\vBold'}(F')$ and $\pi$ induces an injective mapping, which maps blocks of $\NF_{\vBold}(F)$ to the blocks of $\NF_{\vBold'}(F')$. Let now  $\{\xBold'_1,\xBold'_2,\xBold'_3,\xBold'_4\}$ be a block of $\NF_{\vBold'}(F')$. Clearly, the inverse of $\mathcal{L}$ has the form
	\begin{equation*}\label{equation: EA permutation inverse}
		\mathcal{L}^{-1}\colon\graphElt{\xBold}{\yBold} \mapsto \begin{pmatrix}
			\mathbf{A}'_{11} & \mathbf{O}_{\phantom{12}}\\
			\mathbf{A}'_{21} & \mathbf{A}'_{22}\end{pmatrix} \graphElt{\xBold}{\yBold} \oplus \graphElt{\aBold}{\bBold}.
	\end{equation*}
	with $\mathbf{A}'_{11}=\mathbf{A}^{-1}_{11},  \mathbf{A}'_{21}=\mathbf{A}^{-1}_{22}\mathbf{A}_{21}\mathbf{A}_{11}$ and $\mathbf{A}'_{22}=\mathbf{A}^{-1}_{22}$. Let $\xBold_i:=\pi^{-1}(\xBold'_i)=\mathbf{A}^{-1}_{11}\xBold'_i\oplus \aBold$ and $F(\xBold_i)=\mathbf{A}'_{21}\xBold'_i \oplus \mathbf{A}'_{22} F'(\xBold'_i) \oplus \bBold$. In this way, the following equality holds
	\begin{equation*}
		\graphElt{\xBold_1}{F(\xBold_1)}\oplus \graphElt{\xBold_2}{F(\xBold_2)}\oplus \graphElt{\xBold_3}{F(\xBold_3)}\oplus \graphElt{\xBold_4}{F(\xBold_4)}=\graphElt{\mathbf{0}}{\mathbf{A}'_{22}\bigoplus\limits_{i=1}^4F'(\xBold'_i) }=\graphElt{\mathbf{0}}{\mathbf{A}'_{22}\mathbf{v}'}=\graphElt{\mathbf{0}}{\mathbf{v}},
	\end{equation*}
	from what follows that $\{\xBold_1,\xBold_2,\xBold_3,\xBold_4\}$ is a block of $\NF_{\vBold}(F)$. Hence, $\pi$ induces a bijection between the block sets of $\NF_{\vBold}(F)$ and $\NF_{\vBold'}(F')$. Thus the nonvanishing flats $\NF_{\vBold}(F)$ and $\NF_{\vBold'}(F')$ are isomorphic.
\end{proof}

\begin{remark}\label{remark: Nonvanishing flats of APN functions}
	In general, the behavior of the nonvanishing flats is not a CCZ-invariant for $(n,m)$-functions, as the following example shows. We endow $\F_{2}^{6}$ with the structure of the finite field $\left(\F_{2^{6}},+,\cdot \right)$ in such a way, that the multiplicative group $\F^*_{2^{6}}$ is given by $\F^*_{2^{6}}=\langle a \rangle$, where $a$ is a root of the primitive polynomial $p(x)=x^6+x^4+x^3+x+1$. Consider the following CCZ-equivalent but not EA-equivalent functions on $\F_{2^{6}}$: the Kim's APN function $x\in\F_{2^6}\mapsto K(x)$ and the Dillon's APN permutation $x\in\F_{2^6}\mapsto G(x)$, of which the univariate representations can be found in~\cite{APNPermutation}. It is not difficult to check with a computer that for the Kim's APN function there exist:
	\begin{itemize}
		\item $42$ elements $v\in\F_{2^6}$ such that $\NF_{v}(K)$ is a $1$-$(64,4,9)$ design;
		\item $21$ elements $v\in\F_{2^6}$ such that $\NF_{v}(K)$ is a $1$-$(64,4,13)$ design.
	\end{itemize}
	At the same time, among the nonvanishing flats of Dillon's APN permutation only 7 of them, namely $\NF_{v}(G)$ for $v\in V=\{ 1, a^{7}, a^{8}, a^{29}, a^{44}, a^{50}, a^{53} \}$, are $1$-$(64,4,13)$ designs.
\end{remark}
\subsection{Characterization of plateaued functions}\label{subsection 3.2}
As we mentioned in Corollary~\ref{corollary: Coding-theoretic characterization of bent functions}, $(n,m)$-bent functions are those $(n,m)$-functions which have the minimum possible number of the vanishing flats $|\VF_F|$. In this way, the property of the vanishing flats to be a 2-design is in some sense redundant with respect to the characterization of bentness. Further we will show that in contrast to the bent case, one indeed needs the information about all nonvanishing flats in order to characterize the class of plateaued functions. First, we give a formula for the number of vanishing flats for an arbitrary plateaued $(n,m)$-function.
\begin{remark}
	Using Equation~\eqref{equation: Number of weight 4 codewords in CF dual}, it is not difficult to compute the number of vanishing flats for plateaued functions. Let $F$ be an $(n,m)$-plateaued function. For each nonzero $\bBold\in\F_2^m$ let $s_{\bBold}$ be an integer with $0\le s_{\bBold} \le n$ such that the component function $F_{\bBold}$ is $s_{\bBold}$-plateaued. Then the number of vanishing flats $|\mathcal{VF}_{F}|$ is given by
	\begin{equation}\label{equation: Vanishing flats of plateaued functions}
		|\VF_{F}|=\dfrac{1}{3}\left(2^{3n-m-3}+2^{2n-m-3}\sum\limits_{\bBold\in\F_2^m \setminus \{ \mathbf{0} \} }2^{s_{\bBold}}-3\cdot 2^{2n-3} +2^{n-2}  \right).
	\end{equation}
	In particular, if $F$ is an $s$-plateaued $(n,m)$-function, then
	\begin{equation}\label{equation: Number of vanishing flats of s-plateaued functions}
		|\VF_{F}|=\dfrac{1}{3}\left(2^{3n-m-3} + 2^{2n+s-3} - 2^{2n + s-m-3}-3\cdot 2^{2n-3} +2^{n-2}  \right).
	\end{equation}
	\noindent However, this information, namely the number of vanishing flats $|\VF_{F}|$ is not enough to characterize the class of all plateaued functions, as it was shown in~\cite[Section 5.2]{MesnagerOS18}. There exist infinite families of nonplateaued Boolean functions, having the fourth power moments of the Walsh transform of plateaued functions, and, hence the same number of vanishing flats $|\VF_{F}|$, as one can see from Result~\ref{equation: Number of weight 4 codewords in CF dual}. In this way, the same characterization in terms of the number of vanishing flats is no longer possible for $(n,m)$-plateaued functions. 
\end{remark}

Further we will show that one has to analyze the combinatorial structure of nonvanishing flats of $(n,m)$-functions in order to characterize the plateauedness. First, we recall the well-known characterization of plateaued functions by Carlet.
\begin{result}\label{result: Second-order derivatives}\cite[Theorem 1]{Carlet2015}
	Let $F$ be an $(n,m)$-function and $D_{\aBold,\bBold}F(\xBold):=F(\xBold)\oplus F(\xBold\oplus \aBold)\oplus F(\xBold\oplus \bBold)\oplus F(\xBold\oplus \aBold \oplus \bBold)$ be the second-order derivative of $F$ at the point $\xBold\in\F_2^n$. For $\vBold\in\F_2^m$ and $\xBold\in\F_2^m$ we define $N_F(\mathbf{v};\xBold)$ to be the cardinality of the set
	\begin{equation}\label{equation: Set of pairs}
		\left\{(\aBold, \bBold) \in \F_{2}^{n}\times \F_{2}^{n} \colon D_{\aBold,\bBold}F(\xBold)=\vBold \right\}.
	\end{equation}
	Then for any $\xBold\in \F_2^n$ and $\mathbf{v}\in \F_2^m$ the number $N_F(\mathbf{v};\xBold)$ can be computed as follows
	\begin{equation}\label{equation: Almost replication number}
		N_F(\mathbf{v};\xBold) =2^{-m}\cdot \sum_{\mathbf{u} \in \F_2^m} \sum_{\aBold, \bBold \in \F_{2}^{n}}(-1)^{\langle \mathbf{u}, D_{\aBold,\bBold}F(\xBold)\rangle_m \oplus \langle \mathbf{u}, \mathbf{v} \rangle_m }.
	\end{equation}
\end{result}

\begin{theorem}\label{theorem: Characterization of plateaued functions via regularity of VF}
	Let $f$ be a Boolean function on $\F_2^n$. The following statements are equivalent.
	\begin{enumerate}
		\item The function $f$ is $s$-plateaued.
		\item $\VF(f)$ is a $1$-$\left(2^n,4,\dfrac{2^{n+s-1}(2^{n-s}+1)-3\cdot 2^n+2}{6}\right)$ design.
		\item $\NF_1(f)$ is a $1$-$\left(2^n,4,\dfrac{2^{n+s-1}(2^{n-s}-1)}{6}\right)$ design.
	\end{enumerate}
\end{theorem}
\begin{proof}
	
	By~\cite[Theorem
	1]{CarletP03}, a Boolean function $f$ on $\F_{2}^{n}$ is $s$-plateaued if and only if for all $\xBold \in \mathbb{F}_{2}^{n}$ holds
	\begin{equation}
		\sum_{\aBold, \bBold \in \mathbb{F}_{2}^{n}}(-1)^{D_{\aBold,\bBold}f(\xBold)}=2^{n+s}.
	\end{equation}
	In this way, from~\eqref{equation: Almost replication number} we deduce that a Boolean function $f$ on $\F_2^n$ is $s$-plateaued if and only if
	\begin{equation}
		N_f(v;\xBold)=2^{n-1}(2^n+(-1)^v \cdot 2^s)
	\end{equation}
	for all $\xBold \in \mathbb{F}_{2}^{n}$. Clearly, a given point $\xBold\in\F_2^n$ is incident with the block $B=\{ \xBold, \xBold\oplus \aBold, \xBold\oplus \bBold, \xBold\oplus \aBold \oplus \bBold \}$ of $\VF(F)$ if and only if there exists a 2-dimensional vector subspace $\langle \aBold, \bBold \rangle$ with $\aBold, \bBold\in \F_2^n$ such that $D_{\aBold,\bBold}f(\xBold)=0$. In order to determine the number of such two-dimensional vector subspaces it is enough to exclude from the set $\left\{(\aBold, \bBold) \in \F_{2}^{n}\times \F_{2}^{n} \colon D_{\aBold,\bBold}f(\xBold)=v \right\}$
	the pairs $(\aBold,\mathbf{0})$, $(\mathbf{0},\bBold)$ and $(\aBold,\aBold)$, which do not correspond to the affine two-dimensional vector spaces, and divide the cardinality of the obtained set by 6, since any 2-dimensional vector subspace $\langle \aBold, \bBold \rangle$ can be represented by 6 different pairs $(\aBold, \bBold)$ with $\aBold, \bBold\in \F_2^n$. In this way, any point $\xBold\in\F_2^n$ is incident with exactly
	\begin{equation*}
		\lambda_0=\frac{N_f(0;\xBold)-(2\cdot(2^n-1)+2^n)}{6}=\dfrac{2^{n-1}(2^n+2^s)-3\cdot 2^n+2}{6}
	\end{equation*}
	blocks of $\VF(f)$. In the case of the nonvanishing flats $\NF_1(f)$ one applies the same argument as in the case of the vanishing flats, however now without throwing away the pairs of the form $(\aBold,\mathbf{0})$, $(\mathbf{0},\bBold)$ and $(\aBold,\aBold)$, since they lead only to the constant zero (but not constant one) second-order derivatives. In this way, any point $\xBold\in\F_2^n$ is incident with exactly
	\begin{equation*}
		\lambda_1=\frac{N_f(1;\xBold)}{6}=\dfrac{2^{n-1}(2^n-2^s)}{6}
	\end{equation*}
	blocks of $\NF_1(f)$, what, in turn, completes the proof.
\end{proof}
	
As we will show further, a similar characterization is also valid for $(n,m)$-plateaued functions. First we recall the following characterization of vectorial plateaued functions.

\begin{result}\label{result: Characterization of plateaued functions}\cite[Theorem 1]{Carlet2015}.
	Let $F$ be an $(n,m)$-function. Then the following holds.
	\begin{enumerate}
		\item $F$ is plateaued if and only if, for every $\vBold \in \F_{2}^{m}$ the number $N_F(\mathbf{v};\xBold)$
		does not depend on $\xBold \in \F_{2}^{n}$
		\item $F$ is plateaued with single amplitude if and only if the number $N_F(\mathbf{v};\xBold)$ does not depend on $\xBold \in \F_{2}^{n}$, nor on $\vBold \in \F_{2}^{m}$ when $\vBold \neq \mathbf{0}$.
	\end{enumerate}
\end{result}

Following the proof of Theorem~\ref{theorem: Characterization of plateaued functions via regularity of VF}, we characterize plateaued functions as those $(n,m)$-functions, for which the collection $\{ \mathcal{NF}_{\vBold}(F) \colon \vBold\in \F_2^n\setminus \{ \mathbf{0} \} \} \} \sqcup  \{ \VF(F) \}$ is a partition of the affine Steiner quadruple system into 1-designs.

\begin{theorem}\label{theorem: Vanishing flats of plateaued functions}
	Let $F$ be an $(n,m)$-function and $\lambda_{\mathbf{v}}\in \N$ be defined in the following way
	\begin{equation}\label{equation: Regularity of plateaued functions}
		\lambda_\mathbf{0}=\dfrac{N_F(\mathbf{0};\xBold)-3\cdot 2^n+2}{6}\quad \mbox{and}\quad  \lambda_\mathbf{v}=\dfrac{N_F(\mathbf{v};\xBold)}{6} \mbox{ for } \vBold\in\F_2^m\setminus \{ \mathbf{0} \}.
	\end{equation}
	Then the function $F$ is plateaued if and only if for all $\vBold\in\F_2^m\setminus \{ \mathbf{0} \}$ the incidence structure $\NF_{\vBold}(F)$ is a $1$-$(2^n,4,\lambda_{\mathbf{v}})$ design. Additionally, the vanishing flats $\VF(F)$ of an $(n,m)$-function $F$ is a $1$-$(2^n,4,\lambda_{\mathbf{0}})$ design.
\end{theorem}
\begin{remark}
	We want to emphasize, that the information about the regularity of vanishing flats is not necessary, provided that all the nonvanishing flats are regular, since vanishing and nonvanishing flats form a partition of the affine Steiner quadruple system.
\end{remark}
Finally, if a plateaued function has a single amplitude, one can determine the regularity of vanishing and nonvanishing flats explicitly.
\begin{corollary}\label{corollary: Nonvanishing flats of s-plateaued (n,m)-functions}
	Let $F$ be an $(n,m)$-function. The following statements are equivalent.
	\begin{enumerate}
		\item The function $F$ is $s$-plateaued.
		\item For all $\vBold\in\F_2^m\setminus \{ \mathbf{0} \}$ the incidence structure $\NF_{\vBold}(F)$ is a
		\begin{equation*}
			1\mbox{-}\left(2^n,4,\dfrac{2^{n+s-m}(2^{n-s}-1)}{6}\right)
		\end{equation*}
		design. Additionally, the vanishing flats $\VF(F)$ of an $s$-plateaued $(n,m)$-function $F$ is a
		\begin{equation*}
			1\mbox{-}\left(2^n,4,\dfrac{2^{n+s-m}(2^{n-s}+2^m-1)-3\cdot 2^n+2}{6}\right)
		\end{equation*}
		design.
	\end{enumerate}
\end{corollary}	

\paragraph{Group algebra interpretation.}
Group ring equations are a powerful tool, used to characterize nice combinatorial objects. For more background on group ring equations we refer to~\cite{Pott1995FiniteGeometry}, and their applications to perfect nonlinear functions we refer to~\cite{Pott04,Pott16}. For instance, Boolean functions $f$ on $\F_2^n$ with $n=2k$ may be identified with their supports $\support{f}\subseteq\F_2^{n}$, which in the bent case are $\left(2^{n},2^{n-1} \pm 2^{k-1}, 2^{n-2} \pm 2^{k-1}\right)$ difference sets, characterized by the following group ring equation
$$D_f^2= \left(2^{n-1} \pm 2^{k-1}\right)\cdot\mathbbm{1}_G+\left(2^{n-2} \pm 2^{k-1}\right)\cdot G$$
with $G=\F_2^n$. In a similar manner one can identify an $(n,m)$-function $F$ with its graph $\graph{F}\subset G=\F_2^n\times\F_2^m$. In this way, $(n,m)$-bent functions, being $\left(2^{n}, 2^{m}, 2^{n}, 2^{n-m}\right)$-difference sets in $G=\F_2^{n}\times \F_2^{m}$ relative to the subgroup $N=\{(\mathbf{0}, \mathbf{y})\colon \mathbf{y}\in\F_2^m \}$, can be described by the following group ring equation, 
$$\graph{F}^2 = 2^{n} \cdot \mathbbm{1}_G + 2^{n-m}\cdot(G-N).$$
In a similar manner, Budaghyan and Pott~\cite[Theorem 5]{BudaghyanP09} characterized $s$-plateaued $(n,n)$-functions $F$ using the group ring representation of $\graph{F}^3$ in the following way,
\begin{equation}
	\graph{F}^3= 2^{n+s}\cdot \graph{F} + (2^n-2^s)\cdot G.
\end{equation}
Further we generalize this statement by specifying the group ring equation, which a graph of an $(n,m)$-plateaued function has to satisfy. Clearly the following system of equations,
\begin{equation}
	\left\{\begin{array}{l}
		\xBold_1 \oplus \xBold_2 \oplus \xBold_3= \xBold \\
		F(\xBold_1)\oplus F(\xBold_2)\oplus F(\xBold_3)=\yBold
	\end{array}\right.
\end{equation}
is equivalent to the system of equations 
\begin{equation}
	\left\{\begin{array}{l}
		\xBold_1 \oplus \xBold_2 \oplus \xBold_3 \oplus \xBold =\mathbf{0} \\
		F(\xBold_1)\oplus F(\xBold_2)\oplus F(\xBold_3) \oplus F(\xBold) = \vBold
	\end{array}\right.,
\end{equation}
where $\vBold:=F(\xBold) \oplus \yBold$. In this way, for an arbitrary $(n,m)$-function $F$, we obtain that
\begin{equation} 
	\graph{F}^3=\sum_{\vBold\in \F_2^m} \sum_{\xBold\in \F_2^n} \left[ N_F(\vBold,\xBold)\cdot\graphElt{\xBold}{F(\xBold)\oplus \vBold}\right].
\end{equation}
Since the values $N_F(\vBold,\xBold)$ are independent of $\xBold\in\F_2^n$, according to Result~\ref{result: Characterization of plateaued functions}, we get the following characterization of $(n,m)$-plateaued functions.
\begin{theorem}
	An $(n,m)$-function $F$ is plateaued if and only if its graph $\graph{F}$ satisfies the following group ring equation
	\begin{equation}
		\graph{F}^3=\sum_{\vBold\in \F_2^m} N_F(\vBold,\xBold)\cdot\left[\graph{F}\oplus\graphElt{\mathbf{0}}{\vBold}\right].
	\end{equation}
	In particular, an $(n,m)$-function $F$ is $s$-plateaued if and only if its graph $\graph{F}$ satisfies the following group ring equation,
	\begin{equation}
		\graph{F}^3= 2^{n+s}\cdot \graph{F} +(2^{2n-m}-2^{n+s-m})\cdot G,
	\end{equation}
	where $G=\F_2^n\times \F_2^m$.
\end{theorem}

\subsection{Characterization of bent functions among plateaued functions}\label{subsection 3.3}
In this subsection, we show that bent functions are special mappings among plateaued functions, since they partition the affine Steiner quadruple system not only into 1-designs, but also into 2-designs. The proof is based on the connection between bent functions and relative difference sets and follows the proof of Theorem~\ref{theorem: Bent functions via vanishing flats}.
\begin{theorem}\label{theorem: Bent functions via nonvanishing flats}
	Let $F$ be an $(n,m)$-function. The following statements are equivalent.
	\begin{enumerate}
		\item The function $F$ is $(n,m)$-bent.
		\item For any $\vBold\in{\F_2^m}\setminus \{ \mathbf{0} \}$ the incidence structure $\NF_{\vBold}(F)$ is a $2$-$(2^n,4,2^{n-m-1})$ design. 
	\end{enumerate}
	Moreover, the number of the nonvanishing flats of an $(n,m)$-bent function $F$ with respect to a 
	nonzero vector $\vBold\in\F_2^m$ is given by
	\begin{equation}\label{equation: The number of nonvanishing flats of an (n,m)-bent function}
		|\NF_{\vBold,F}|=\frac{\left(2^{n+m}-2^m\right) \cdot 2^{2(n-m)}}{24}.
	\end{equation}
\end{theorem}

\begin{proof}
	\emph{1.$\Rightarrow$2.} Let $F$ be an $(n,m)$-bent function, $\vBold$ be a nonzero element of $\F_2^m$ and $\NF_{\vBold}(F)=(\mathcal{P},\mathcal{B})$ be the nonvanishing flats of $F$ with respect to $\vBold$. We will show that any two different points $\xBold_1,\xBold_2$ of $\mathcal{P}$ are contained in exactly $2^{n-m-1}$ blocks of $\mathcal{B}$. We define $\aBold:=\xBold_1\oplus\xBold_2$ and let $\vBold':=F(\xBold_1)\oplus F(\xBold_2)$ and $\vBold'':=\vBold'\oplus \vBold$. Then the following holds
	\begin{equation}
		\graphElt{\xBold_1}{F(\xBold_1)}\oplus \graphElt{\xBold_2}{F(\xBold_2)}=\graphElt{\aBold}{\vBold'}=\graphElt{\mathbf{0}}{\vBold}\oplus \graphElt{\aBold}{\vBold''}.
	\end{equation}
	Since the graph $\graph{F}$ is a $\left(2^{n}, 2^{m}, 2^{n}, 2^{n-m}\right)$-difference set in the group $G=\F_2^{n} \times \F_2^{m}$ relative to the forbidden subgroup $N=\{(\mathbf{0}, \mathbf{y})\colon \mathbf{y}\in\F_2^m \}$, the element  $g:=\graphElt{\aBold}{\vBold''} \in G\setminus N$ has $2^{n-m}$ representations
	\begin{equation}
		g=\graphElt{\xBold_3}{F(\xBold_3)}\oplus \graphElt{\xBold_4}{F(\xBold_4)}
	\end{equation}
	with $\{ \xBold_3,\xBold_4 \} \neq \{ \xBold_1,\xBold_2 \}$. In this way, any 2-subset $ \{ \xBold_1,\xBold_2 \}$ is contained in exactly $2^{n-m-1}$ blocks $\{ \xBold_1,\xBold_2,\xBold_3,\xBold_4 \}$ of the form \eqref{equation: Block set of the nonvanishing flats},	from what follows that $\NF_{\vBold}(F)$ is a $2$-$(2^n,4,2^{n-m-1})$ design. The count of the blocks is the same as in Remark~\ref{remark: Number of vanishing flats of a bent function}.
	
	\noindent\emph{2.$\Rightarrow$1.} Follows from the fact that vanishing and nonvanishing flats form a partition of the affine Steiner quadruple system, and that the resulting number of vanishing flats is equal to the value in~\eqref{equation: Bentness via CF dual}, which is the minimum possible value among all $(n,m)$-functions.
\end{proof}

\section{Characterizations of extendable and lonely bent functions}\label{section 5}
Many questions about Boolean functions have a coding theoretic interpretation, for instance, as it is well-known, the question about the maximum nonlinearity of a Boolean function is equivalent to the covering radius problem of the  \emph{first-order Reed-Muller code} $\mathcal{RM}(n,1)$. In this section we show, that the concept of extendable (lonely) $(n,m)$-bent functions $F$ can be interpreted in terms of the covering radius problem for the linear code $\mathcal{C}_F$. Additionally we provide a purely combinatorial description of the extendability problem of bent functions by means of the subdesign problem, using the theory of vanishing and nonvanishing flats, developed in previous sections.
\subsection{Coding-theoretic approach}\label{subsection 5.1}
For an $(n,m)$-bent function $F$,
consider the $[2^n, n+m+1 ,2 ^{n-1} - 2^{n/2-1}]$-linear code $\mathcal{C}_F\subset \F_2^{2^n}$, which contains the first-order Reed-Muller code $\mathcal{RM}(n,1)$. As it is well known, the covering radius $\rho$ of $\mathcal{RM}(n,1)$ is $\rho = 2^{n-1}-2^{n/2-1}$, when $n$ is even. The definition of the covering radius can be extended for an arbitrary subset of $\F_2^v$ in the following way, see \cite{Oblaukhov2019,Oblaukhov2021} for details.
\begin{definition}
	The \emph{covering radius} $\rho=\rho(A)$ of the subset $A\subseteq\F_2^v$ is  $\rho=\max\limits_{\xBold\in\F_2^v}\min\limits_{\aBold\in A}d(\xBold, \aBold)$. The set $\widehat{A}=\{\xBold \in\F_2^v \colon d(\xBold,A)=\rho(A)\}$ is called the \emph{metric complement} of $A$. If $\widehat{\widehat{A}} = A$, then $A$ is called a \emph{metrically regular set}.
\end{definition}

Clearly, metrically regular sets always come in pairs, as with $A$ also $\widehat{A}$ is metrically regular. Bent functions are exactly the functions at maximal distance $\rho = 2^{n-1}-2^{n/2-1}$ from the set of affine 
functions. Hence, by definition,
the set $\mathcal{B}_n$ of Boolean bent functions in (even) dimension $n$ is the metric complement of $\mathcal{RM}(n,1)$. As it was shown in \cite{Tok12}, $\mathcal{RM}(n,1)$, $n$ even, (and $\mathcal{B}_n$) are metrically regular sets, i.e., we also have $\widehat{\mathcal{B}_n} = \mathcal{RM}(n,1)$. Since the code $\mathcal{C}_F$ of an $(n,m)$-bent function $F$ contains $\mathcal{RM}(n,1)$ as a subcode, the covering radius $\rho=\rho(\mathcal{C}_F)$ can be at most $2^{n-1}-2^{n/2-1}$.
As we will show, the covering radius of $\mathcal{C}_F$ is strongly connected with the concept of lonely bent functions and the concept of extendable bent functions, see Definition~\ref{definition: lonely bent functions}. 

\begin{remark}
	By the Nyberg's bound \cite{Nyberg91}, a vectorial $(n,n/2)$-bent function is non-extendable. However, so far there are no known examples of non-extendable $(n,m)$-bent functions with $m<n/2$. Moreover,  as the recent study~\cite{polujan2020design} shows, all $(6,m)$-bent functions with $m=1,2$ are extendable.
\end{remark}
The question on the existence of lonely Boolean bent functions has a connection with \emph{the bent sum decomposition problem}, formulated by Tokareva~\cite[Hypothesis 1]{Tokareva2011}, which if it holds, also excludes the existence of lonely Boolean bent functions.
\begin{hypothesis}\label{hypothesis: Tokarevas}\cite[Hypothesis 1]{Tokareva2011}
	Every Boolean function of algebraic degree at most $n/2$ in $n$ variables, can be expressed as the sum of two bent functions in $n$ variables ($n$ is even, $n\ge2$).
\end{hypothesis}

Though it is only confirmed for a few classical constructions of Boolean functions that they are a sum of two bent functions (mostly seen straightforward, see~\cite{QuFDL14}), the hypothesis is not falsified. In fact if this were true in general, then one would have quite precise (upper and lower) bounds for the number of Boolean bent functions in any dimension $n$.

Let $f,g$ be two Boolean functions on $\F_2^n$, then $f\oplus g$ is bent if and only if $d(f\oplus g,l_\aBold) = 
2^{n-1}-2^{n/2-1}$ or $d(f\oplus g,l_\aBold) = 2^{n-1}+2^{n/2-1}$ 
(then $d(f\oplus g,l_\aBold\oplus 1) = 2^{n-1}-2^{n/2-1}$). We can state this fact also as follows.

\begin{lemma}\label{lemma: Sum is bent trivially}
	Let $f,g$ be two Boolean functions on $\F_2^n$, which are not necessarily bent. Then their sum $h=f\oplus g$ is bent if and only if $d(f\oplus l_\aBold, g) = 2^{n-1}\pm 2^{n/2-1}$ for all $\aBold\in\F_2^n$.
\end{lemma}

The following lemma shows that we can relax the condition in Lemma \ref{lemma: Sum is bent trivially}. We will use this fact to show a connection between loneliness (extendability) of bent functions and the covering radii of their codes.

\begin{lemma}\label{lemma: Sufficient condition}
	Let $f,g$ be two Boolean functions on $\F_2^n$, which satisfy the following inequality
	\begin{equation}
	2^{n-1}-2^{n/2-1} \le d(f\oplus l_\aBold,g) \le 2^{n-1}+2^{n/2-1}
	\end{equation}
	for all $\aBold\in\F_2^n$. Then we have $d(f\oplus l_\aBold,g)=2^{n-1}\pm 2^{n/2-1}$ for all $\aBold\in\F_2^n$.
\end{lemma}
\begin{proof}
	We prove the lemma by contradiction. Recall that by Parseval's identity, for every Boolean function $h$ on $\F_2^n$ we have $\sum_{\aBold\in\F_2^n}\hat{\chi}_h(\aBold)^2 = 2^{2n}$. Suppose that for all $\aBold\in\F_2^n$ we have 
	$d(f\oplus l_\aBold,g) = 2^{n-1}-2^{n/2-1}+\epsilon_\aBold$ with $0 \le \epsilon_\aBold \le 2^{n/2}$, and suppose that for 
	at least one $\tilde{\aBold}$ we have $0 < \epsilon_{\tilde{\aBold}} < 2^{n/2}$. Then 
	\begin{align*}
	\hat{\chi}_{f\oplus g}(\aBold) & = \sum_{\xBold\in\F_2^n}(-1)^{g(\xBold)\oplus f(\xBold)\oplus \aBold\cdot \xBold} 
	= 2^n - 2d(f\oplus l_\aBold,g) \\
	& = 2^n - 2(2^{n-1}-2^{n/2-1}+\epsilon_\aBold) = 2^{n/2} - 2\epsilon_\aBold.
	\end{align*}
	With $0 \le \epsilon_\aBold \le 2^{n/2}$ we have $\hat{\chi}_{f\oplus g}(\aBold)^2 = (2^{n/2} - 2\epsilon_\aBold)^2 \le 2^n$ with equality if and only if $\epsilon_\aBold = 0$ or $\epsilon_\aBold = 2^{n/2}$. However by assumption, for 
	$\tilde{\aBold}$ we then have $\hat{\chi}_{f\oplus g}(\tilde{\aBold})^2 < 2^n$. This contradicts Parseval's identity for the function $f\oplus g$ on $\F_2^n$. 
\end{proof}
	Lemma~\ref{lemma: Sum is bent trivially} and Lemma~\ref{lemma: Sufficient condition} can be trivially extended to the vectorial bent case, by considering the necessary and sufficient conditions component-wise.
\begin{corollary}\label{corollary: Sufficient condition vectorial}
	Let $n$ be even and $F,G$ be two $(n,m)$-functions, which are not necessarily bent. The following statements are equivalent.
	\begin{enumerate}
		\item The sum $F\oplus G$ is an  $(n,m)$-bent function.
		\item For all nonzero $\vBold\in\F_2^m$ the equality $d(F_\vBold\oplus l_\aBold, G_\vBold) = 2^{n-1}\pm 2^{n/2-1}$ holds for all $\aBold\in\F_2^n$.
		\item For all nonzero $\vBold\in\F_2^m$ the inequality $2^{n-1}-2^{n/2-1} \le d(F_\vBold\oplus l_\aBold, G_\vBold) \le 2^{n-1}+2^{n/2-1}$ holds for all $\aBold\in\F_2^n$.
	\end{enumerate}
\end{corollary}

In the following statement, we provide a coding-theoretic characterization of extendable and lonely bent functions.

\begin{theorem}\label{theorem: Lonely bent functions}
	Let $F$ be an $(n,m)$-bent function with $m\le n/2-1$. Then $F$ is extendable if and only if the linear code $\mathcal{C}_F$ has the covering radius $\rho(\mathcal{C}_F)=2^{n-1}-2^{n/2-1}$. The metric complement $\widehat{\mathcal{C}_F}$ of $\mathcal{C}_F$ for an extendable $(n,m)$-bent function $F$ is
	\begin{equation}\label{equation: metric component}
		\widehat{\mathcal{C}_F} = \{f\in\mathcal{B}_n\;:\;f\oplus F_{\vBold}\;\mbox{is bent for all }\,\vBold\in\F_2^m \}.
	\end{equation}
\end{theorem}
\begin{proof}
	%Certainly, as we add vectors to $\mathcal{RM}(n,1)$, the covering radius can only decrease, i.e. $\rho(\mathcal{C}_F)$ is at most $2^{n-1}-2^{n/2-1}$.
	Let $f$ be a Boolean function on $\F_2^n$ and $\tilde{F}$ be an $(n,m+1)$-function, defined as follows $\tilde{F}\colon\xBold\mapsto (F(\xBold),f(\xBold))$. Clearly, for any such a function $\tilde{F}$ the following inequality holds,
	$$\rho(\mathcal{C}_{\tilde{F}})\le \rho(\mathcal{C}_{F})=2^{n-1}-2^{n/2-1}$$ with equality if only if $f\not\in \mathcal{C}_F$ is bent (so that the distance from $\mathcal{RM}(n,1)$ is kept) and satisfies $d(f, F_{\vBold}\oplus l) \ge 2^{n-1}-2^{n/2-1}$ for all 
	$l \in\mathcal{RM}(n,1)$ and $\vBold\in\F_2^m$.
	Note that then $f$ also satisfies $d(f, F_{\vBold}\oplus l) \le 2^{n-1}+2^{n/2-1}$ for all 
	$l \in\mathcal{RM}(n,1)$ and $\vBold\in\F_2^m$. 
	With Corollary~\ref{corollary: Sufficient condition vectorial} we then conclude, that $F$ is not extendable if and only if $\rho(\mathcal{C}_F) < 2^{n-1}-2^{n/2-1}$. Finally, the  metric complement $\widehat{\mathcal{C}_F}$ is given by~\eqref{equation: metric component}.
\end{proof}
\begin{remark}
	Alternatively one can describe the metric complement of $A=\mathcal{C}_F$ for an extendable $(n,m)$-bent function $F$ as follows
	\[ \widehat{A} = \{f \colon f\oplus g \mbox{ is bent for all } g \in A \}. \]
	This description of the metric complement then also applies to $A=\mathcal{RM}(n,1)$. Finally we would like to mention, that for $(n,m)$-bent functions $F$ with $n=6$ the values $|\widehat{\mathcal{C}_F}|/2^{n+1}$ were studied under the name \emph{number of bent friends} in~\cite{polujan2020design}, however outside the context of metric properties of Boolean and vectorial functions.
\end{remark}

\begin{remark}
	In the light of the Theorem \ref{theorem: Lonely bent functions}, the question about the existence of Boolean bent functions $f$ on $\F_2^n$ satisfying $\rho(\mathcal{C}_{f})<\rho(\mathcal{RM}(n,1))$, is the same as the existence question of lonely Boolean bent functions.
\end{remark}

From Lemma \ref{lemma: Sufficient condition} or Theorem \ref{theorem: Lonely bent functions} we can straightforwardly infer some more statements, which may shed a light on covering radius and metric complement of the codes $\mathcal{C}_F$, and on the existence of lonely and extendable bent functions. Further we denote by $B(f, r):=\{ g\in\mathfrak{B}_n \colon d(f, g) < r \}$ the \emph{open Hamming ball} of radius $r$ centered at the Boolean function $f$ on $\F_2^n$.
\begin{corollary}\label{corollary: Bent and covering radius}
	Let $\rho = 2^{n-1}-2^{n/2-1}$ be the covering radius of $\mathcal{RM}(n,1)$ and $\mathcal{B}_n$ be the set of Boolean bent functions on $\F_2^n$. The following hold.
	\begin{enumerate}
		\item A Boolean function $f$ on $\F_2^n$ is the sum of two bent functions on $\F_2^n$ if and only if
		\begin{equation*}
		\mathcal{B}_n\not\subset \bigcup\limits_{l \in\mathcal{RM}(n,1)} B(f\oplus l, \rho).
		\end{equation*}
		\item An $(n,m)$-bent function $F$ is non-extendable (lonely) if and only if
		\begin{equation*}
		\mathcal{B}_n \subset \bigcup\limits_{\vBold \in \F_2^m,l \in\mathcal{RM}(n,1)} B(F_\vBold \oplus l, \rho).
		\end{equation*}
		\item Let $n=2m$ and $F$ be a vectorial $(n,m)$-bent function given as
		$F( \xBold) = (f_1( \xBold),\ldots, f_m( \xBold))^T$. We define the $(n,m-1)$-bent function $F'$ as follows $F'( \xBold):= (f_1( \xBold),\ldots, f_{m-1}( \xBold))^T
		$. Then $\mathcal{C}_{F'}$ is a linear $[2^n,3n/2,2^{n-1}-2^{n/2-1}]$-code with covering radius $2^{n-1}-2^{n/2-1}$. The metric complement $\widehat{\mathcal{C}_{F'}}$ of $\mathcal{C}_{F'}$ contains 
		$\{f_m \oplus l \colon l\in \mathcal{RM}(n,1) \}$. This set is exactly the metric complement if and only if $F$ is the only way to extend $F'$ to an $(n,m)$-bent function
		(up to the addition of an affine function to the component $f_m$).
		\item Let $F=(f_1(\xBold),\ldots,f_m(\xBold))^T$ be an $(n,m)$-bent function, which is extendable by a Boolean bent function $f$ on $\F_2^n$. Then $F$ is extendable by any bent function $f'$ on $\F_2^n$ from the set $f\oplus\langle f_1,\ldots f_m \rangle$. In general the metric complement of the linear code $\mathcal{C}_{F}$ for an extendable $(n,m)$-bent function $F(\xBold)=(f_1(\xBold),\ldots,f_m(\xBold))^T$ has the following structure
			\begin{equation}\label{equation: Structure of metric component}
				\widehat{\mathcal{C}_F}=\left(\bigsqcup_{f'}  f' \oplus \langle f_1,\ldots f_m \rangle\right) \oplus\  \mathcal{RM}(n,1),
			\end{equation}
		where $f'$ runs through all different coset leaders, which extend $F$.
	\end{enumerate}
\end{corollary}

\begin{remark}\label{remark: Uniquely extendable bent functions}
	From the classification of vectorial bent functions in six variables~\cite[Figure IV.1.]{polujan2020design} and the structure of the metric complement~\eqref{equation: Structure of metric component}, one can see that the only uniquely extendable (up to the choice of the coset leader, having no affine terms in ANF) class of vectorial bent functions in 6 variables, is the class $C^2_1$ given in~\cite[Figure IV.1.]{polujan2020design}.
\end{remark}

By the definition of bent function and covering radius of $\mathcal{RM}(n,1)$, $n$ even, the union of the open Hamming balls $B(l,\rho)$ of radius $\rho=\rho(\mathcal{RM}(n,1))$ and centers $l\in\mathcal{RM}(n,1)$ contains all Boolean functions but the bent functions. Further we observe that Boolean functions of algebraic degree larger than $n/2$ have completely different behavior.

\begin{corollary}
	\label{allbentin}
	Let $\rho = 2^{n-1}-2^{n/2-1}$ and let $f$ be a Boolean function on $\F_2^n$ with $\deg(f)>n/2$. Then the union of the Hamming balls $B(f \oplus l, \rho)$ with radius $\rho$ and centers $f \oplus l$ for $l\in\mathcal{RM}(n,1)$ contains all 
	bent functions, i.e.,
	\begin{equation*}
	\mathcal{B}_n \subset \bigcup\limits_{l \in\mathcal{RM}(n,1)} B(f \oplus l, \rho).
	\end{equation*}
\end{corollary}
\begin{proof}
	First, recall that a bent function can have algebraic degree at most $n/2$, see~\cite{ROTHAUS1976300}. Consequently, there cannot exist a bent function $g$ such that $g\oplus f$ is again bent, i.e., $f$ is not the sum of two bent functions. With the first part of Corollary \ref{corollary: Bent and covering radius}, the claim follows.
\end{proof}
Note that the collection of the Hamming balls $B(f\oplus l,\rho)$ with $l \in \mathcal{RM}(n,1)$ is the collection of Hamming balls $B(l,\rho)$ with the centers $l \in \mathcal{RM}(n,1)$ shifted by $f$.
\begin{remark}
If the algebraic degree of a Boolean function $f$ on $\F_2^n$ is $\deg(f)\le1$ respectively $\deg(f)\ge n/2+1$, then the union of the open Hamming balls $B(f\oplus l,\rho)$ with $l \in \mathcal{RM}(n,1)$ and $\rho=\rho(\mathcal{RM}(n,1))$, contains no respectively all bent functions. For a Boolean function $f$ on $\F_2^n$ of algebraic degree $d$ with $2\le d\le n/2$, we only know that there must be some bent function in the union of the corresponding Hamming balls. With Corollary \ref{corollary: Bent and covering radius}, we can formulate some equivalent statements.
\end{remark}
\begin{corollary}
	The following statements are equivalent.
	\begin{enumerate}
		\item There exists no lonely Boolean bent function on $\F_2^n$.
		\item For every Boolean bent function $f$ on $\F_2^n$ we have 
		$\mathcal{B}_n\not\subset \bigcup\limits_{l \in\mathcal{RM}(n,1)} B(f\oplus l, \rho)$.
	\end{enumerate}
\end{corollary}
\begin{proof}
	Follows from the second part of Corollary \ref{corollary: Bent and covering radius} applied to Boolean bent functions.
\end{proof}
In a similar manner one can describe the bent sum decomposition problem. 
\begin{corollary}
The following statements are equivalent.
\begin{enumerate}
	\item Tokareva's Hypothesis~\ref{hypothesis: Tokarevas} holds.
	\item For every Boolean function $f$ on $\F_2^n$ the inclusion $\mathcal{B}_n\subset \bigcup\limits_{l \in\mathcal{RM}(n,1)} B(f\oplus l, \rho)$ holds if and only if $\deg(f)>n/2$.
\end{enumerate}
\begin{proof}
By the first part of Corollary~\ref{allbentin}, Tokareva's Hypothesis~\ref{hypothesis: Tokarevas} holds
is and only if for every Boolean function $f$ of degree at most $n/2$ we have
$\mathcal{B}_n\not\subset \bigcup\limits_{l \in\mathcal{RM}(n,1)} B(f\oplus l, \rho)$.
Since a Boolean functions of degree larger than $n/2$ cannot be the sum of two bent functions,
again by Corollary~\ref{allbentin} we have 
$\mathcal{B}_n\subset \bigcup\limits_{l \in\mathcal{RM}(n,1)} B(f\oplus l, \rho)$ for all Boolean functions
with degree larger than $n/2$. This completes the proof.
\end{proof}
\end{corollary}

As it was mentioned in~\cite{Tok12}, the first-order Reed-Muller code $\mathcal{RM}(n,1)$ for $n$ even is a metrically regular set. This motivates the  following question.
\begin{question}\label{question: Metric regularity}
	Can the set $\mathcal{C}_F$ for some $(n,m)$-bent function $F$ be metrically regular?
\end{question}
As remarked above, by the Nyberg's bound the covering radius of $\mathcal{C}_F$ is smaller than $2^{n-1}-2^{n/2-1}$ if $F$ is an $(n,m)$-bent function with $m=n/2$.
\begin{example}
	Up to equivalence there are exactly two $(4, m)$-bent functions: one Boolean bent function $f$ and one vectorial $(4,2)$ bent function $F$, which algebraic normal forms are given below
	\begin{equation*}
		f(\xBold):=x_1 x_2 \oplus x_3 x_4\quad \mbox{and} \quad F(\xBold):=\begin{pmatrix}
		x_1 x_2 \oplus x_3 x_4 \\ x_1 x_2  \oplus x_1 x_4 \oplus x_2 x_3
		\end{pmatrix}.
	\end{equation*}
	Since $f$ is a coordinate function of $F$, it is extendable, and thus by Theorem~\ref{theorem: Lonely bent functions} we have that $\rho(\mathcal{C}_f)=\rho(\mathcal{RM}(4,1))=6$. However, since $F$ achieves the Nyberg's bound, it is lonely and hence, $\rho(\mathcal{C}_F)<\rho(\mathcal{C}_f)$. With a Magma~\cite{MR1484478} program it is not difficult to check, that $\rho(\mathcal{C}_F)=4$.
\end{example}
\begin{question}\label{question: Covering radius problem}
	Let $F$ be an $(n,n/2)$-bent function. What is the covering radius of $\mathcal{C}_F$? Does it depend on $F$? More general, how much can the covering radius of $\mathcal{C}_F$ for a non-extendable $(n,m)$-bent function $F$ drop?
\end{question}

We finish this section on the connection between properties of $\mathcal{C}_F$ and the extendability of a  (vectorial) bent function $F$ with a more general observation. Let $A$ be a subset of $\F_2^v$ containing $ \oBold$, and $B$ its metric complement, with distance $\rho$,  i.e., the covering radius of $A$ is $\rho$. For a subset $U\subset B$, let $AU:= \{ \aBold \oplus \lambda \uBold\,:\,  \aBold\in A, \uBold\in U, \lambda\in\F_2\}$.

\begin{theorem}
	The set $AU$ has covering radius $\rho$ if and only if there exists $\wBold\in B\setminus U$ such that 
	$w\oplus U \subset B$. The metric complement of $AU$ is then 
	$\widehat{AU} = \{\wBold\in B\setminus U\,:\,\wBold\oplus U\subset B\}$.
	\begin{proof}
		The covering radius of $AU$ is $\rho$ if and only if there exists $\wBold \in B$ (to keep the distance from $A$) 
		such that $d( \aBold\oplus\uBold,\wBold)\ge \rho$ for all $\uBold\in U$ and $ \aBold\in A$.
		Suppose that for some $\wBold\in B$ we have $\uBold\oplus\wBold\in B$ for all $\uBold\in U$. Note that since 
		$\oBold\in A$ we have $\wBold\not\in U$. Then for all $\aBold\in A$ and $\uBold\in U$ we have
		\begin{equation}\label{distance}
		d( \aBold\oplus\uBold,\wBold) = \wt( \aBold\oplus\uBold\oplus\wBold) = 
		d( \aBold,\uBold\oplus\wBold) \ge \rho.
		\end{equation} 
		(Note that in particular for $ \aBold=\oBold$ we get $d(\uBold,\wBold)\ge \rho$.) Conversely suppose that $AU$ has covering radius $\rho$, i.e., there exists $\wBold\in B$ such that that $d( \aBold\oplus\uBold,\wBold)\ge \rho$ for all $\uBold\in U$ and $ \aBold\in A$. Then by $(\ref{distance})$ we have $d( \aBold,\uBold\oplus\wBold) \ge \rho$, i.e., $\uBold\oplus\wBold\in B$ for all $\uBold\in U$ (note that $\wBold\not\in U$ since $ \oBold\not\in B$). By definition, the metric complement of $AU$ contains all $\wBold\in B$ with $d( \aBold\oplus\uBold,\wBold) \ge \rho$ for all $ \aBold\in A$, $\uBold\in U$. By the above arguments this are exactly all $\wBold \in B$	for which $\wBold\oplus U\subset B$.
	\end{proof}
\end{theorem}

	\subsection{Design-theoretic approach}\label{subsection 5.2}
	In this subsection, we use nonvanishing flats to show that vanishing flats of extendable bent functions must be highly structured combinatorial objects, from what follows, that the existence of certain subdesigns in $\VF(F)$ is a measure of extendability of an $(n,m)$-bent function $F$.

	For a vectorial $(n,m)$-bent function $F(\xBold) = (f_1(\xBold),\ldots,f_m(\xBold))^T$ we define the  projections $F_s(\xBold):= (f_1(\xBold),\ldots,f_s(\xBold))^T$ and $F_{m-s}(\xBold):=(f_{s+1}(\xBold),\ldots,f_m(\xBold))^T$, which are $(n,s)$- and  $(n,m-s)$-bent functions, respectively.
	Let $\mathcal{VF}_F$, $\mathcal{VF}_{F_s}$ and $\mathcal{VF}_{F_{m-s}}$ be the block sets of the  vanishing flats of $F$, $F_s$ and $F_{m-s}$, respectively, and let
	\begin{equation*}
		\mathcal{NF}_{F}: = \bigsqcup_{\vBold\in\F_2^m\setminus\{\mathbf{0}\}}\mathcal{NF}_{\vBold,F}
	\end{equation*}
	be the (disjoint) union of the block sets of the nonvanishing flats of the vectorial bent function $F$. Define the collection of the vanishing flats of the projection $F_s$, disjoint from the vanishing flats of $F_{m-s}$ in the following way
	\begin{equation*}
		\mathcal{DF}_{F/F_{m-s}}: = \mathcal{VF}_{F_s} \cap \mathcal{NF}_{F_{m-s}}.
	\end{equation*}
	We can associate the incidence structures with the defined collections of 4-subsets as follows
	$$\mathcal{DF}(F/F_{m-s}):= (\F_2^n,\mathcal{DF}_{F/F_{m-s}}) \quad \mbox{and} \quad \mathcal{NF}(F) := (\F_2^n,\mathcal{NF}_{F}).$$
	In order to analyze the combinatorial structure of the introduced incidence structures we need the following lemma, which summarizes the relations between parameters of designs having no common blocks.
	\begin{lemma}\label{obvious}
		Let $D_1 = (\mathcal{P},\mathcal{B}_1)$ and $D_2 = (\mathcal{P},\mathcal{B}_2)$ be $t$-$(v,k,\lambda_1)$
		and $t$-$(v,k,\lambda_2)$ designs, respectively, and suppose that $\mathcal{B}_1\cap\mathcal{B}_2 = \varnothing$.
		\begin{enumerate}
			\item Then
			$D = (\mathcal{P},\mathcal{B}_1\cup\mathcal{B}_2)$ is a $t$-$(v,k,\lambda_1+\lambda_2)$ design.
			\item Conversely, if $D_1=(\mathcal{P},\mathcal{B}_1)$ is a $t$-$(v,k,\lambda_1)$ design and 
			$D = (\mathcal{P},\mathcal{B}_1\cup\mathcal{B}_2)$ is a $t$-$(v,k,\lambda_1+\lambda_2)$ design,
			then $D_2$ is a $t$-$(v,k,\lambda_2)$ subdesign of $D$.
		\end{enumerate}
	\end{lemma}
	Using the previous lemma we show that the incidence structures $\mathcal{DF}(F/F_{m-s})$ and $\mathcal{NF}(F)$ of vectorial bent functions are 2-designs and determine their parameters.
	\begin{proposition}
		\label{sumdes}
		For a vectorial $(n,m)$-bent function $F(\xBold) = (f_1(\xBold),\ldots,f_m(\xBold))^T$, consider the projections
		$F_s(\xBold) = (f_1(\xBold),\ldots,f_s(\xBold))^T$ and $F_{m-s}(\xBold) =(f_{s+1}(\xBold),\ldots,f_m(\xBold))^T$. With the notation above, the following hold.
		\begin{enumerate}
			\item $\mathcal{VF}_F = \mathcal{VF}_{F_s} \cap \mathcal{VF}_{F_{m-s}}$ and	$\mathcal{VF}_{F_s} = \mathcal{VF}_F \sqcup \mathcal{DF}_{F/F_{m-s}}$, where $\sqcup$ denotes a disjoint union.
			\item $\mathcal{DF}(F/F_{m-s})$ is a $2$-$(2^n,4,(2^{m-s}-1)\cdot 2^{n-m-1})$ design.
			\item $\mathcal{NF}(F)$ is a $2$-$(2^n,4,(2^m-1)\cdot 2^{n-m-1})$ design.
		\end{enumerate}
	\end{proposition}
	\begin{proof}
		Clearly, $\{ \xBold_1,\xBold_2,\xBold_3,\xBold_4 \}$ is a vanishing flat of $F$ if and only if it is a vanishing flat of both, $F_s$ and $F_{m-s}$, hence $\mathcal{VF}_F = \mathcal{VF}_{F_s} \cap \mathcal{VF}_{F_{m-s}}$. On the other hand, for $F_s$ seen as a projection of $F$ we can distinguish two kinds of vanishing flats, those which are also in $\mathcal{VF}_F$ (note that $\mathcal{VF}_F\subset\mathcal{VF}_{F_s}$), and those which are in $\mathcal{VF}_{F_s}$ but not in $\mathcal{VF}_{F_{m-s}}$. The latter is exactly the set $\mathcal{DF}_{F/F_{m-s}}$, and we obtain that $\mathcal{VF}_{F_s} = \mathcal{VF}_F \sqcup \mathcal{DF}_{F/F_{m-s}}$. With Lemma \ref{obvious}, $\mathcal{DF}(F/F_{m-s})$ is then a $2$-$(2^n,4,(2^{m-s}-1)\cdot 2^{n-m-1})$ design. Finally, $\mathcal{NF}(F)$ is a $2$-$(2^n,4,(2^m-1)\cdot 2^{n-m-1})$ design by Lemma \ref{obvious}.
	\end{proof}
	By Proposition \ref{sumdes}, the $2$-$(2^n,4,2^{n-s-1}-1)$ design $\VF(F)$ of a (vectorial) $(n,s)$-bent function $F$, which is a projection of a vectorial $(n,m)$-bent function $\tilde{F}$ for some $m>s$, has certain structural properties. With these observations, we can give a connection between extendability of a (vectorial) bent function and vanishing flats.
	\begin{corollary}\label{ex-able}
		Let $F$ be an $(n,s)$-bent function.
		\begin{enumerate}
			\item If $F$ is extendable, then there exist subdesigns $D=(\F_2^n,\mathcal{B})$ and $D_1=(\F_2^n,\mathcal{B}_1)$  of $\VF(F)$ with parameters $2$-$(2^n,4,2^{n-s-2}-1)$ and $2$-$(2^n,4,2^{n-s-2})$, respectively, such that $\VF(F) = D \sqcup D_1$.
			\item If $F$ is a projection of a vectorial $(n,s+r)$-bent function $\tilde{F}$ for some $s+r\le n/2$, then there exists a partition $$\VF(F)=D\sqcup \left( \bigsqcup_{i=1}^r D_i \right),$$ where $D=\VF(\tilde{F})$ is a $2$-$(2^n,4,2^{n-s-r-1}-1)$ design and for all $1 \le i \le r$ the incidence structures $D_i=(\F_2^n,\mathcal{B}_i)$ are $2$-$(2^n,4,2^{n-s-1-i})$ designs with  $|\mathcal{B}_i|= (2^{3n-s-3-i} - 2^{2n-s-3-i})/3$.
		\end{enumerate}
	\end{corollary}
	\begin{proof}
		Note that the first statement follows as the special case $r=1$ from the second. Let $F$ and $\tilde{F}$ be given as $F(\xBold) = (f_1(\xBold),\ldots,f_s(\xBold))^T$ and $\tilde{F}(\xBold) = (f_1(\xBold),\ldots, f_s(\xBold),f_{s+1}(\xBold),\ldots,f_{s+r}(\xBold))^T$. By Proposition \ref{sumdes}, the block set $\mathcal{VF}_F$ of the vanishing flats $\VF(F)$ of $F$, seen as a projection of the function $F_1(\xBold) = (f_1(\xBold),\ldots, f_s(\xBold),f_{s+1}(\xBold))^T$, is a (disjoint) union of the vanishing flats $\mathcal{VF}_{F_1}$ of $F_1$ and the set $\mathcal{B}_1 = \mathcal{DF}_{F_1/f_{s+1}}$ of cardinality $|\mathcal{B}_1|=|\mathcal{VF}_F| - |\mathcal{VF}_{F_1}| = (2^{3n-s-4} - 2^{2n-s-4})$. Again by Proposition \ref{sumdes}, and by Lemma \ref{obvious}, $D_1=(\F_2^n,\mathcal{B}_1)$ is a $2$-$(2^n,4,2^{n-s-2})$ design.	As $F_1$ is a projection of $F_2(\xBold) = (f_1(\xBold),\ldots,f_{s+1}(\xBold),f_{s+2}(\xBold))^T$, $\mathcal{VF}_{F_1}$ is a (disjoint) union of $\mathcal{VF}_{F_2}$ and $\mathcal{B}_2 = \mathcal{DF}_{F_2/f_{s+2}}$. The cardinality of $\mathcal{B}_2$ is given by $|\mathcal{B}_2|=|\mathcal{VF}_{F_1}| - |\mathcal{VF}_{F_2}| = (2^{3n-s-5} - 2^{2n-s-5})$, and again, $D_2=(\F_2^n,\mathcal{B}_2)$ is a $2$-$(2^n,4,2^{n-s-3})$ design. With a recursive argument, the second statement is shown.
	\end{proof}
	Finally, we apply Corollary \ref{ex-able} to derive a sufficient condition for lonely respectively non-extendable bent functions, which may potentially be used for a computer search of non-extendable bent functions based on the subdesign problem, which can be efficiently solved using the DESIGN package~\cite{Soicher19} of a system for computational discrete algebra GAP~\cite{GAP4}. For an example of use of  the DESIGN package for solving subdesign problems we refer to~\cite{Soicher2013}.
	
	\begin{corollary}\label{corollary: Search of nonextendable bents}
		Let $F$ be an $(n,s)$-bent function.
		\begin{enumerate}
			\item If $\VF(F)$ contains no $2$-$(2^n,4,2^{n-s-2}-1)$ subdesign (or no $2$-$(2^n,4,2^{n-s-2})$ subdesign), then $F$ is not extendable.
			\item More general, if $\VF(F)$ contains no $2$-$(2^n,4,2^{n-s-r-1}-1)$ subdesign for some $r$ satisfying $1\le r\le n/2-s-1$, then $F$ is not the projection of an $(n,n/2)$-bent function.
		\end{enumerate}
	\end{corollary}
	\section{Conclusion and open problems}\label{section 6}
	In this paper, we considered design-theoretic and coding-theoretic aspects of cryptographically significant Boolean and vectorial functions. As a highlight we completed the characterization of one of the most important classes of cryptographic functions, namely, differentially 2-valued and plateaued functions in terms of the incidence structures. For instance, Li et al.~\cite{Li20VanishingFlats} and Tang, Ding and Xiong~\cite{TangDing2019} showed, that the valency of vanishing flats reflects differential uniformity, from what follows that differentially two-valued $(n,n)$-functions $F$ can be characterized in terms of vanishing flats $\VF(F)$ having the property to be 2-designs. In this paper, we showed that regularity of nonvanishing flats reflects another important cryptographic property, namely, plateauedness, and consequently we derived new characterizations of $(n,m)$-plateaued and $(n,m)$-bent functions in terms of nonvanishing flats having the property to be 1-designs and 2-designs, respectively.
	
	In  Table~\ref{table: Bent functions and their generalizations from design-theoretic point of view}, we summarize various design-theoretic characterizations of cryptographically significant classes of $(n,m)$-functions and mention, what kind of incidence structures one gets from the supports of codewords of a fixed weight. We denote by ``$\iff$'' a condition or combination of conditions, which characterizes a certain class of $(n,m)$-functions, and by ``$\Longrightarrow$'' the properties of the supported incidence structures  of a certain class of $(n,m)$-functions. One may observe a remarkable property of bent functions: all three constructions of incidence structures (vanishing flats, nonvanishing flats and supports of the codewords of a fixed weight) always lead to 2-designs. This is, in general, not the case for differentially two-valued $(n,n)$-functions and $(n,m)$-plateaued functions: one can see from Table~\ref{table: Bent functions and their generalizations from design-theoretic point of view} which combinatorial properties of a bent function one may lose, if one considers various generalizations.
	
	\begin{table}[h]
		\caption{Bent functions and their generalizations from design-theoretic point of view}
		\label{table: Bent functions and their generalizations from design-theoretic point of view}
		\centering
		\scalebox{.97}{\begin{tabular}{|>{\centering\arraybackslash}m{.22\linewidth}||>{\centering\arraybackslash}m{.21\linewidth}>{\centering\arraybackslash}m{.22\linewidth}|>{\centering\arraybackslash}m{.23\linewidth}|} \hline 
				Classes of $(n,m)$-functions $F$ & \multicolumn{1}{c|}{Vanishing flats $\VF(F)$} & Nonvanishing flats $\NF_{\vBold}(F)$ &  Supports $\left(\mathcal{P}\left(\mathcal{C}_F\right), \mathcal{B}_{\ell}\left(\mathcal{C}_F\right)\right)$ and $\left(\mathcal{P}\left(\mathcal{C}_F^{\perp}\right), \mathcal{B}_{k>4}\left(\mathcal{C}_F^{\perp}\right)\right)$\\ \hline \hline
				\multirow{3}{*}{$(n,m)$-Bent functions} & \multicolumn{1}{c|}{2-design}                & 2-designs                            & 2-designs \\
				& \multicolumn{1}{c|}{$\iff$}                   & $\iff$                               & $\Longrightarrow$    \\
				& \multicolumn{1}{c|}{By Theorem~\ref{theorem: Bent functions via vanishing flats}}                 & By~Theorem~\ref{theorem: Bent functions via nonvanishing flats}                             & By~\cite[Example 4]{TangDing2019}   \\ \hline
				Differentially two-valued      & 2-design                & Equiregular 1-designs                            & 2-designs \\
				$s$-plateaued                         & \multicolumn{2}{c|}{$\quad \ \iff$}                               & $\Longrightarrow$    \\
				$(n,n)$-functions              & \multicolumn{2}{c|}{By~\cite[Theorem 6.1]{TangDing2019} and Corollary~\ref{corollary: Nonvanishing flats of s-plateaued (n,m)-functions} }                              & By~\cite[Theorem 6.4]{TangDing2019}    \\ \hline
				Differentially       & \multicolumn{1}{c|}{2-design}                & Not necessarily                             & \multirow{3}{*}{TBD$^*$} \\
				two-valued                         & \multicolumn{1}{c|}{$\iff$} &     1-designs                          &     \\
				$(n,n)$-functions              & \multicolumn{1}{c|}{By~\cite[Theorem 6.1]{TangDing2019}}                  & By Remark~\ref{remark: Nonvanishing flats of APN functions}                              &     \\ \cline{1-3} \hline
				\multirow{2}{*}{$s$-Plateaued} & \multicolumn{2}{c|}{Nonvanishing flats are equiregular 1-designs}                            & \multirow{3}{*}{TBD$^{**}$} \\
				& \multicolumn{2}{c|}{$\quad \ \iff$}                               &     \\
				$(n,m)$-functions&  \multicolumn{2}{c|}{By Corollary~\ref{corollary: Nonvanishing flats of s-plateaued (n,m)-functions}}                              &    \\ \cline{1-3} \hline
				\multirow{2}{*}{Plateaued} & \multicolumn{2}{c|}{Nonvanishing flats are 1-designs}                            & \multirow{3}{*}{TBD$^{**}$} \\
				& \multicolumn{2}{c|}{$\quad \ \iff$}                               &     \\
				$(n,m)$-functions& \multicolumn{2}{c|}{By Theorem~\ref{theorem: Vanishing flats of plateaued functions}}                              &    \\ \cline{1-3} \hline
			\end{tabular}}
	\end{table}	
	
	It is out of the scope of this paper to answer all the questions about cryptographic functions, their incidence structures and linear codes. Finally, we would like to give a list of questions and open problems, which, we think, deserve further investigations.
	\begin{enumerate}
		\item What are the incidence structures, supported by the codewords of a fixed weight arising from differentially two-valued and plateaued functions, marked by TBD$^{*}$ and TBD$^{**}$ (to be determined) in Table~\ref{table: Bent functions and their generalizations from design-theoretic point of view}? We give some further insights.
		\begin{enumerate}[label*=\arabic*.]
			\item First we consider the TBD$^{*}$ entry. In general it may be difficult to say under which conditions the linear codes $\mathcal{C}_F$ and $\mathcal{C}_F^\perp$ of differentially two-valued $(n,n)$-functions $F$ support 2-designs. In Theorem~\ref{theorem: APNs with classical Walsh spectrum support 2-designs}, we specified one such a condition for APN functions, namely, having the classical Walsh spectrum. It would be interesting to find out, whether it is an if and only if condition, or to find more classes of APN functions supporting 2-designs.
			\item Now we consider TBD$^{**}$ entries. Deleting a coordinate function from quadratic APN functions from Dillon's list, one may get 1-designs from the obtained projections. However, it is not clear theoretically, why it happens, since the extended Assmus-Mattson Theorem is not applicable any more. In this way, a more careful analysis of this case is needed, although we do not expect that one can get interesting incidence structures out of this construction (we expect at most 1-designs).
		\end{enumerate}
%		\item The next problem is related to a potential characterization of quadratic bent functions among the set of all bent functions. Based on our computer experiments with Boolean bent functions on $\F_2^{n}$ with $n\le 8$, we conjecture that quadratic Boolean bent functions can be characterized by 2-transitivity of one of the following automorphism groups: $\Aut(\D(f))$, $\Aut(\VF(f))$, $\Aut(\mathcal{C}_f)$ and $\Aut(\mathcal{C}_f^\perp)$.
	\end{enumerate}	
	The following series of questions is related to extendability of bent functions and metric regularity of their linear codes.
	\begin{enumerate}
	\setcounter{enumi}{2}
		\item Based on Corollary~\ref{corollary: Search of nonextendable bents} one may try to develop an algorithm for the search of non-extendable bent functions. At the same time one may also try to find theoretically the complexity of the extendability problem. As a starting point we refer to~\cite{COLBOURN198559}, where complexity of the subdesign problem is established for several classes of designs and subdesigns.
		\item Let $F$ be an $(n,s)$-bent function. By Corollary~\ref{corollary: Search of nonextendable bents}, the non-existence of a $2$-$(2^n,4,2^{n-s-r-1}-1)$ subdesign of a $2$-$(2^n,4,2^{n-s-1}-1)$ design $\VF(F)$ gives information about non-extendability of $F$. Moreover, do $2$-$(2^n,4,2^{n-s-1}-1)$ designs without $2$-$(2^n,4,2^{n-s-r-1}-1)$ subdesigns exist? On the other hand, assume one can find (computationally) a $2$-$(2^n,4,2^{n-s-r-1}-1)$ subdesign $D$ of $\VF(F)$ for an $(n,s)$-bent function $F$. Can this design $D$ be realized as $\VF(\tilde{F})$ of an $(n,s+r)$ function $\tilde{F}$, which contains $F$ as a projection?
		\item As we mentioned in Remark~\ref{remark: Uniquely extendable bent functions}, uniquely extendable bent functions (up to the choice of the coset leader, having no affine terms in ANF) are exceptionally rare in a small number of variables. In this way, we think it would be interesting to find out, which $(n,n/2)$-bent functions can be characterized by the property to be ``unique'' extensions of $(n,n/2-1)$-bent functions.
		\item  For the first-order Reed-Muller code $\mathcal{RM}(n,1)$, $n$ even, weight distribution, covering radius are well-known and its metric regularity was established recently, see~\cite{Tok12}.
		The weight distribution of linear codes $\mathcal{C}_F$ for $(n,m)$-bent functions $F$ is known as well as the covering radius $\rho(\mathcal{C}_F)$ for extendable $(n,m)$-bent functions $F$, see Theorem~\ref{theorem: Lonely bent functions}. However, it is still not clear, whether linear codes $\mathcal{C}_F$ of extendable $(n,m)$-bent functions $F$ can be metrically regular. Moreover, in the case of non-extendable $(n,m)$-bent functions $F$ the problem seems to be more complicated, since the covering radius $\rho(\mathcal{C}_F)$ is not known, see also Questions~\ref{question: Metric regularity} and~\ref{question: Covering radius problem}.
		\end{enumerate}
		%\newpage
		\section*{Acknowledgement}
		Wilfried Meidl is supported by FWF projects P 30966, P 35138.
		
		\providecommand{\bysame}{\leavevmode\hbox to3em{\hrulefill}\thinspace}
		\providecommand{\MR}{\relax\ifhmode\unskip\space\fi MR }
		\providecommand{\MRhref}[2]{%
			\href{http://www.ams.org/mathscinet-getitem?mr=#1}{#2}
		}
		\providecommand{\href}[2]{#2}

\end{document}